\documentclass[envcountsame]{llncs} 
\usepackage[pdfpagelabels, colorlinks=true,urlcolor=webblue,linkcolor=webgreen,filecolor=webblue,citecolor=webgreen,pdfpagemode=UseOutlines,pdfstartview=FitH,pdfpagelayout=OneColumn,bookmarks=true,hyperfootnotes=false]{hyperref}

\usepackage{etex}
\usepackage{graphicx}
\usepackage[english]{babel}
\usepackage[utf8]{inputenc}
\usepackage{ifthen}
\usepackage{url}
\usepackage{booktabs}
\usepackage{xspace}
\usepackage{xcolor}

\usepackage{amsmath}
\usepackage{amssymb}
\usepackage{amsbsy}
\usepackage{dsfont}
\usepackage{float}
\usepackage{braket}

\usepackage{cleveref}
\usepackage{tikz}
\usetikzlibrary{shapes, arrows, positioning}
\usepackage[all]{xy}

\usepackage{color}
\definecolor{webgreen}{rgb}{0,.5,0}
\definecolor{webblue}{rgb}{0,0,.5}

\tikzstyle{block} = [draw, rectangle, minimum height=2em, minimum width=3em]
\tikzstyle{input} = [rectangle]
\tikzstyle{output} = [coordinate]
\tikzstyle{trace} = [draw, isosceles triangle]

\newboolean{SUBMISSION}
\newboolean{EXTENDED}
\setboolean{SUBMISSION}{false}
\newboolean{LNCS}
\setboolean{LNCS}{true}
\newcommand{\comment}[1]{}   

\newcommand{\A}{\mathcal{A}}

\newcommand{\C}{\mathcal{C}}
\newcommand{\D}{\mathcal{D}}

\renewcommand{\H}{\mathcal{H}}
\newcommand{\K}{\mathcal{K}}
\newcommand{\M}{\mathcal{M}}

\renewcommand{\P}{\mathcal{P}}

\renewcommand{\S}{\ensuremath{\mathcal{S}}}

\newcommand{\NN}{\mathbb{N}}

\newcommand{\poly}{\ensuremath{\text{poly}}}
\newcommand{\rand}{\raisebox{-1pt}{\ensuremath{\,\xleftarrow{\raisebox{-1pt}{$\scriptscriptstyle\$$}}}}}
\newcommand{\from}{\ensuremath{\leftarrow}}

\newtheorem{scheme}{Scheme}

\newcommand{\one}{\mathds 1}

\newcommand{\outerprod}[2]{|#1\rangle\langle #2|}

\newcommand\supp{\textbf{supp\,}}

\DeclareMathOperator{\tr}{Tr}

\newcommand{\expref}[2]{\texorpdfstring{\hyperref[#2]{#1~\ref{#2}}}{#1~\ref{#2}}}

\newcommand{\revise}[1]{}

\newcommand{\eps}{\varepsilon}

\newcommand{\alg}{\mathcal}

\newcommand{\KeyGen}{\ensuremath{\mathsf{KeyGen}}\xspace}
\newcommand{\Enc}{\ensuremath{\mathsf{Enc}}\xspace}
\newcommand{\Dec}{\ensuremath{\mathsf{Dec}}\xspace}

\newcommand{\negl}{\operatorname{negl}}
\newcommand{\inrand}{\rand}
\newcommand{\pr}{\operatorname{Pr}}

\newcommand{\bit}[1]{\{0,1\}^{#1}}

\newcommand{\reg}[1]{\ensuremath{#1}}

\newcommand{\regF}{\reg{F}}

\newcommand{\ketbra}[2]{\ket{#1}\bra{#2}}			
\newcommand{\egoketbra}[1]{\ketbra{#1}{#1}}			

\newcommand{\states}{\mathfrak D}

\newcommand{\theoremStatementqOWF}{If quantum-secure one-way functions exist, then so do IND-CCA1-secure private-key quantum encryption schemes.}
\newcommand{\theoremStatementqTOWP}{If quantum-secure trapdoor one-way permutations exist, then so do semantically secure public-key quantum encryption schemes.}

\begin{document}

\setcounter{tocdepth}{2}
\makeatletter
\renewcommand*\l@author[2]{}
\renewcommand*\l@title[2]{}
\makeatletter

\title{Computational Security of\\ Quantum Encryption}

\ifthenelse{\boolean{SUBMISSION}}
{
\author{\vspace*{-1.5cm} }
\institute{\vspace*{-1cm}\ }
}
{
\author{Gorjan Alagic\inst{1}
	\and Anne Broadbent\inst{2}
	\and Bill Fefferman\inst{3}
	\and Tommaso Gagliardoni \inst{4}
	\and Christian Schaffner\inst{5}
	\and Michael St.~Jules\inst{6}
	}
\institute{
{Department of Mathematical Sciences, University of Copenhagen\\{\tt galagic@gmail.com}} \and
{Department of Mathematics and Statistics, University of Ottawa\\{\tt abroadbe@uottawa.ca}} \and
{Joint Center for Quantum Information and Computer Science (QuICS), University of Maryland.\\{\tt wjf@umd.edu}} \and
{Cryptoplexity, TU Darmstadt, Germany\\{\tt tommaso@gagliardoni.net}} \and
{QuSoft, University of Amsterdam and CWI, The Netherlands\\{\tt c.schaffner@uva.nl}} \and
{Department of Mathematics and Statistics, University of Ottawa\\{\tt mstju032@uottawa.ca}}
}
}

\maketitle
\thispagestyle{plain} 
\begin{abstract}
Quantum-mechanical devices have the potential to transform cryptography. Most research in this area has focused either on the information-theoretic advantages of quantum protocols or on the security of classical cryptographic schemes against quantum attacks. In this work, we initiate the study of another relevant topic: the encryption of quantum data in the computational setting. 

\vspace{.25cm}
In this direction, we establish quantum versions of several fundamental classical results. First, we develop natural definitions for private-key and public-key encryption schemes for quantum data.
We then define notions of semantic security and indistinguishability, and, in analogy with the classical work of Goldwasser and Micali, show that these notions are equivalent. Finally, we construct secure quantum encryption schemes from basic primitives. In particular, we show that quantum-secure one-way functions imply IND-CCA1-secure symmetric-key quantum encryption, and that quantum-secure trapdoor one-way permutations imply semantically-secure public-key quantum encryption.


\end{abstract}


\ifthenelse{\boolean{SUBMISSION}}
{
}
{
\tableofcontents
	\newpage
\setcounter{page}{1}
}


\section{Introduction}
\label{sec:introduction}

Quantum mechanics changes our view of information processing: the ability to access, operate and transmit data according to the laws of quantum physics opens the doors to a vast realm of possible applications.
Cryptography is one of the areas that is most seriously impacted by the potential of quantum information processing, since the security of most cryptographic primitives in use today relies on the hardness of computational problems that are easily broken by adversaries having access to a quantum computer~\cite{Sho94}.

While the impact of quantum computers on cryptanalysis is tremendous, quantum mechanics itself predicts physical phenomena that can be exploited in order to achieve new levels of security. These advantages were already mentioned in the late 1970's in pioneering work of Wiesner~\cite{Wie1983}, and have led to the very successful theory of quantum key distribution (QKD)~\cite{BennettB84}, which has already seen real-world applications~\cite{ABB+2014}. QKD achieves information-theoretically secure key expansion, and has the advantage of relatively simple hardware requirements (notwithstanding a long history of successful attacks to QKD at the implementation level~\cite{ABB+2014}).

The cryptographic possibilities of quantum information go well beyond QKD. Indeed, quantum copy-protection~\cite{Aar2009}, quantum money \cite{Wie1983,AC2012,MS2010} and revocable time-release encryption~\cite{Unr2014} are just some examples where properties unique to quantum data enable new cryptographic constructions (see \cite{BS16} for a survey). 
Thanks in part to these tremendous cryptographic opportunities, we  envisage an increasing need for an information infrastructure that enables quantum information. Such an infrastructure will be required to support:
\begin{itemize}
\item  \textbf{Quantum functionality:} honest parties can store, exchange, and compute on quantum data; \item  \textbf{Quantum security:} quantum functionality is protected against quantum adversaries.
\end{itemize}

The current state-of-the-art is lacking even the most basic cryptographic concepts in the context of quantum functionality and quantum adversaries. In particular, the study of encryption of quantum data (which is arguably one of the most fundamental building blocks) has so far been almost exclusively limited to the quantum one-time pad~\cite{AMTW00} and other aspects of the information-theoretic setting~\cite{D09,DD10} (one notable exception being~\cite{BJ15}). The achievability of other basic primitives such as public-key encryption has not been thoroughly investigated for the case of fully quantum cryptography. This situation leaves many open questions about what can be achieved in the quantum~world.

\subsection{Summary of Contributions and Techniques}
\label{sec:summ-contributions}
In this work, we establish quantum versions of several fundamental classical (\emph{i.e.} ``non-quantum'') results in the setting of computational security. Following Broadbent and Jeffrey~\cite{BJ15}, we consider private-key and public-key encryption schemes for quantum data. In these schemes, the key is a classical bitstring\footnote{While quantum keys might be of interest, they are not necessary for constructing secure schemes~\cite{BJ15}.}, but both the plaintext and the ciphertext are quantum states. Key generation, encryption, and decryption are implemented by polynomial-time quantum algorithms. Such schemes admit an appropriate definition of indistinguishability security, following the classical approach~\cite{BJ15}: the quantum adversary is given access to an encryption oracle, and must output a challenge plaintext; given either the corresponding ciphertext or the encryption of $\outerprod{0}{0}$ (each with probability 1/2), the adversary must decide which was the case.

Our main contributions are the following. First, we give several natural formulations of semantic security for quantum encryption schemes, and show that all of them are equivalent to indistinguishability. This cements the intuition that possession of the ciphertext should not help the adversary in computing anything about the plaintext. Second, we give two constructions of encryption schemes with semantic security: a private-key scheme, and a public-key scheme. The private-key scheme  satisfies a stronger notion of security: indistinguishability against chosen ciphertext attacks (IND-CCA1). A more detailed summary of these contributions follows.

\subsubsection{Semantic Security vs. Indistinguishability}
\label{sec:intro-SEM}

 Semantic security formalizes the notion of security of an encryption scheme under computational assumptions. Originally introduced by Goldwasser and Micali~\cite{GM},  this definition posits a game: an adversary is given the encryption of a message~$x$ and some side information $h(x)$, and is challenged to output the value of an objective function~$f$ evaluated at~$x$. An encryption scheme is deemed secure if every adversary can be closely approximated by a \emph{simulator} who is given only $h(x)$; crucially, the simulator must work for every possible choice $(h, f)$ of side information and objective function.  This models the intuitive notion that having access to a ciphertext gives the adversary essentially no advantage in computing functions related to the plaintext.

While semantic security corresponds to a notion of security that is intuitively strong, it is cumbersome to use in terms of security proofs. In order to address this problem, Goldwasser and Micali~\cite{GM} showed the equivalence of semantic security with another cryptographic notion, called \emph{indistinguishability}. The intuitive description of indistinguishability is also in terms of a game, this time with a \emph{single} adversary. The adversary prepares a pair of plaintexts $x_0$ and $x_1$ and submits them to a challenger, who chooses a uniformly random bit $b$ and returns the encryption of~$x_b$. The adversary then performs a computation and outputs a bit $v$; the adversary wins the game if $v=b$ and loses otherwise. An encryption scheme is deemed secure if no adversary wins the game with probability significantly larger than $1/2$. This definition models the intuitive notion that the ciphertexts are indistinguishable: whatever the adversary does with one ciphertext, the outcome is essentially the same if run on the other ciphertext.

In \expref{Section}{sec:quantum-sem}, we define semantic security for the encryption of \emph{quantum} data---thus establishing a parallel with the notions and results of encryptions as laid out by Goldwasser and Micali. When attempting to transfer the definition of semantic security to the quantum world, the main question one encounters is to determine the quantum equivalents of $h(x)$ and $f(x)$ as described above (because of the no-cloning theorem~\cite{Wootters1982},  we cannot postulate a polynomial-time experiment that simultaneously involves some quantum plaintext \emph{and} a function of the plaintext---see \expref{Section}{sec:defining-semantic} for further discussions related to this issue). We propose a number of alternative definitions in order to deal with this situation (\expref{Definition}{def:SEM}, \expref{Definition}{def:SEMtwo}, and \expref{Definition}{def:SEMthree}.) Perhaps the most surprising is our definition of SEM (\expref{Definition}{def:SEM}), which does away completely with the need to explicitly define  analogues of the functions $h$ and $f$, instead relying on a \emph{message generator} that outputs three registers, consisting of the ``plaintext'', ``side information'' and ``target output'' (there is no further structure imposed on the contents of these registers). Intuitively, we think of the adversary's goal being to output the value contained in the ``target output'' register. Formally, however, 
 \expref{Definition}{def:SEM} shows that the role of the ``target output'' register is actually to help the distinguisher: semantic security corresponding to the situation where no distinguisher has a non-negligible advantage in telling apart the real scenario (involving the adversary) and the ideal scenario (involving the simulator), \emph{even given access to the ``target output'' system}.
Our main result in this direction (see \expref{Section}{sec:IND-equiv-SEM}) is the equivalence between semantic security and indistinguishability for quantum encryption schemes:
\begin{theorem}\label{thm:intro:ind-equiv-sem}
A quantum encryption scheme is semantically secure if and only if it has indistinguishable encryptions.
\end{theorem}
What is more, because our definitions and proofs hold when restricted to the classical case (and in fact can be shown as generalizations of the standard classical definitions), our contribution  sheds new light on semantic security: to the best of our knowledge, this is the first time that semantic security has been defined \emph{without} the need to explicitly refer to functions $h$ and~$f$.

\subsubsection{Quantum Encryption Schemes}
In \expref{Section}{sec:Constructions}, we give two constructions of quantum encryption schemes that achieve semantic security (and thus also indistinguishability, by \expref{Theorem}{thm:intro:ind-equiv-sem}.) Our constructions make use of two basic primitives. The first is a \emph{quantum-secure one-way function} (qOWF). This is a family of deterministic functions which are efficiently computable in classical polynomial time, but which are impossible to invert even in quantum polynomial time. It is believed that such functions can be constructed from certain algebraic problems~\cite{MRV07,KK07}. The existence of qOWFs implies the existence of \emph{quantum-secure pseudorandom functions} (qPRFs)~\cite{Zhandry2012}. We show that a qPRF can, in turn, be used to securely encrypt quantum data with classical private keys. More precisely, we have the following:

\begin{theorem}
\label{thm:intro:SKE}
\theoremStatementqOWF
\end{theorem}

The second basic primitive we consider is a \emph{quantum-secure one-way permutation with trapdoors} (qTOWP). In analogy with the classical case, a qTOWP is a qOWF with an additional property: each function in the family is a permutation whose efficient inversion is possible if one possesses a secret string (the trapdoor). While our results appear to be the first to consider applications to quantum data, the notion of quantum security for trapdoor permutations is of obvious relevance in the security of classical cryptosystems against quantum attacks. Some promising candidate qTOWPs from lattice problems are known~\cite{PW08,GPV08}. We show that such functions can be used to give secure public-key encryption schemes for quantum data, again using only classical keys.
\begin{theorem}
\label{thm:intro:PKE}
\theoremStatementqTOWP
\end{theorem}

We remark that \expref{Theorem}{thm:intro:SKE} and \expref{Theorem}{thm:intro:PKE} are analogues of standard results in the classical literature~\cite{Goldreich2004}.

\subsection{Related Work}

Prior work has considered the computational security of quantum methods to encrypt classical data~\cite{OKS00,sim2,sim3}. Information-theoretic security for the encryption of quantum states has been considered in the context of the one-time pad~\cite{AMTW00,BR2003,HLSW2004,L02}, as well as entropic security~\cite{D09,DD10}. Computational indistinguishability notions for encryption in a quantum world were proposed in independent and concurrent work~\cite{BJ15,Gagliardoni2015}. While~\cite{BJ15} considers the encryption of quantum data (and proposes the first constructions based on hybrid classical-quantum encryption),~\cite{Gagliardoni2015} considers
the security of {\em classical} schemes which can be accessed in a quantum way by the adversary.

The results of~\cite{Gagliardoni2015} are part of a line of research of \emph{``post-quantum''} cryptography, which investigates the security of classical schemes against quantum adversaries, with the goal of finding ``quantum-safe'' schemes. This includes the study of encryption and signature schemes secure against attacks by quantum algorithms~\cite{BBD09}, and also the study of superposition attacks against quantum oracles \cite{BDF+11,Zhandry2012,Unr15}. Still in the model of superposition attacks,~\cite{Boneh2013a} studies quantum indistinguishability under chosen plaintext and chosen ciphertext attacks. This definition was improved in~\cite{Gagliardoni2015} to allow for a quantum challenge phase. The latter paper also initiates the study of quantum semantic security of classical schemes and gives the first classical construction of a quantumly secure encryption scheme from a family of quantum-secure pseudorandom permutations. Another quantum indistinguishability notion in the same spirit has been suggested (but not further analyzed) in~\cite[Def.~5.3]{Velema13}.

Several previous works have considered how classical security proofs change in the setting of quantum attacks (see, e.g., \cite{Unr2010,FKS+2013,Son2014}.) Our results can be viewed as part of this line of work; one distinguishing feature is that we are able to extend classical security proofs to the setting of quantum functionality secure against quantum adversaries. This setting has seen increasing interest in the past decade, with progress being made on several topics: multi-party quantum computation~\cite{BCG+2006}, secure function evaluation~\cite{DNS2010,DNS2012}, one-time programs~\cite{BGS2013},  and delegated quantum computation~\cite{BFK2009,Bro2015}.

\paragraph{\textbf{Outline.}} The remainder of the paper is structured as follows. In \expref{Section}{sec:preliminaries}, we set down basic notation and recall a few standard facts regarding classical and quantum computation. In \expref{Section}{sec:encryption}, we define symmetric-key and public-key encryption for quantum states (henceforth ``quantum encryption schemes''), as well as a notion of indistinguishability (including IND-CPA and IND-CCA1) for such schemes. \expref{Section}{sec:quantum-sem} defines semantic security for
  quantum encryption schemes, and shows equivalence with  indistinguishability. \expref{Section}{sec:Constructions} gives our two constructions for quantum encryption schemes. Finally, we close with some discussion of future work in \expref{Section}{sec:conclusions}.


\section{Preliminaries}\label{sec:preliminaries}

We introduce some basic notation for classical (\expref{Section}{sec:classical-prelims}) and quantum (\expref{Section}{sec:quantum-prelims}) information processing and information-theoretic encryption. \expref{Section}{sec:efficient} concerns basic issues in efficient algorithms and \expref{Section}{sec:oracles} discusses the use of oracles.

\subsection{Classical States, Maps, and the One-Time Pad}
\label{sec:classical-prelims}

Let $\NN$ be the set of positive integers. For $n \in \NN$, we set $[n] = \{1, \cdots, n\}$. Define 
$\bit{*} := \cup_n \bit{n}$. An element $x \in \bit{*}$ is called a bitstring, and $|x|$ denotes its length, \emph{i.e.}, its number of bits. We reserve the notation $0^n$ (resp., $1^n$) to denote the $n$-bit string with all zeroes (resp., all ones).

For a finite set $X$, the notation $x \inrand X$ indicates that $x$ is selected uniformly at random from~$X$. For a probability distribution $S$, the notation $x \from S$ indicates that $x$ is sampled according to $S$. Given finite sets $X$ and $Y$, the set of all functions from $Y$ to $X$ is denoted $X^Y$ (or sometimes $\{X \rightarrow Y\}$).
We will usually consider functions $f$ acting on binary strings, that is, of the form $f: \bit{n} \rightarrow \bit{m}$, for some positive integers $n$ and~$m$.
We will also consider function families  $f:\bit{*} \rightarrow \bit{*}$ defined on bitstrings of arbitrary size. One can construct such a family simply by choosing one function with input size $n$, for each $n$. We will sometimes abuse notation by stating that $f:\bit{n} \rightarrow \bit{m}$ defines a function family; in that case, it is implicit that $n$ is a parameter that indexes the input size and $m$ is some function of $n$ (usually a polynomial) that indexes the output size. Given a bitstring $y$ and a function family $f$, the preimage of $f$ under $y$ is defined by $f^{-1}(y) := \{ x \in \bit{*} : f(x) = y\}$.

We will often write $\negl(\cdot)$ to denote a function from $\NN$ to $\NN$ which is ``negligible'' in the sense that it grows at an inverse-superpolynomial rate. More precisely, $\negl(n) < 1 / p(n)$ for every polynomial $p : \NN \rightarrow \NN$ and all sufficiently large $n$. A typical use of negligible functions is to indicate that the probability of success of some algorithm is too small to be amplified to a constant by a feasible (\emph{i.e.}, polynomial) number of repetitions.

Given two bitstrings $x$ and $y$ of equal length, we denote their bitwise XOR by $x \oplus y$.  Recall that the \emph{classical one-time pad} encrypts a plaintext $x \in \bit{n}$ by XORing it with a uniformly random string (the key) $r \inrand \bit{n}$. Decryption is performed by repeating the operation, \emph{i.e.}, by XORing the key with the ciphertext. Since the uniform distribution on $\bit{n}$ is invariant under XOR by $x$, the ciphertext is uniformly random to parties having no knowledge about~$r$~\cite{Shannon1949}. A significant drawback of the one-time pad is the key length. In order to reduce the key length, one may generate~$r$ pseudorandomly; this key-length reduction requires making computational assumptions about the adversary.

\subsection{Quantum States, Maps, and the One-Time Pad}
\label{sec:quantum-prelims}

Given an $n$-bit string $x$, the corresponding quantum-computational $n$-qubit basis state is denoted~$\ket{x}$. The $2^n$-dimensional Hilbert space spanned by $n$-qubit basis states will be denoted
$$
\H_n := \textbf{span} \left\{ \ket{x} : x \in \bit{n} \right\}\,.
$$
We denote by $\states(\H_n)$ the set of density operators (i.e., valid quantum states) on~$\H_n$. These are linear operators on $\states(\H_n)$ which are positive-semidefinite and have trace equal to $1$. When considering different physical subsystems, we will denote them with uppercase Latin letters; when a Hilbert space corresponds to a subsystem, we will place the subsystem label in the subscript. For instance, if $F \cup G \cup H = [n]$ then $\H_n = \H_F \otimes \H_G \otimes \H_H.$ Sometimes we will write explicitly the subsystems a state belongs to as subscripts; this will be useful when considering, \emph{e.g.}, the reduced state on some of the subspaces. For example, we will sometimes express the statement $\rho \in \states(\H_F \otimes \H_G \otimes \H_H)$ simply by calling the state $\rho_{FGH}$; in that case, the state obtained by tracing out the subsystem~$H$ will be denoted~$\rho_{FG}$.

Given $\rho, \sigma \in \states(\H)$, the trace distance between $\rho$ and $\sigma$ is given by half the trace norm $\|\rho - \sigma\|_1$ of their difference. When $\rho$ and $\sigma$ are classical probability distributions, the trace distance reduces to the total variation distance. Physically realizable maps from a state space $\states(\H)$ to another state space $\states(\H')$ are called \emph{admissible}---these are the completely positive trace-preserving (CPTP) maps. For the purpose of distinguishability via input/output operations, the appropriate norm for CPTP maps is the diamond norm, denoted $\|\cdot\|_\diamond$. The set of admissible maps coincides with the set of all maps realizable by composing (i.) addition of ancillas, (ii.) unitary evolutions, (iii.) measurements in the computational basis, and (iv.) tracing out subspaces. We remark that unitaries $U \in U(\H_n)$ act on $\states(\H_n)$ by conjugation: $\rho \mapsto U \rho U^\dagger$. The identity operator~$\one_n \in U(\H_n)$ is thus both a valid map, and (when normalized by $2^{-n}$) a valid state in $\states(\H_n)$---corresponding to the classical uniform distribution.

Recall  the single-qubit Pauli operators  defined as:
$$
I =
\begin{pmatrix}
1 & 0 \\
0 & 1 \\
\end{pmatrix}\,,
\qquad
X =
\begin{pmatrix}
0 & 1 \\
1 & 0 \\
\end{pmatrix}\,,
\qquad
Y =
\begin{pmatrix}
0 & -i \\
i & 0 \\
\end{pmatrix}\,,
\qquad
Z =
\begin{pmatrix}
1 & 0 \\
0 & -1 \\
\end{pmatrix}\,.
$$
The Pauli operators are Hermitian and unitary quantum gates, i.e.\ $P^\dag=P$ and $P^\dag P=P P^\dag = P^2 = I$ for all $P \in \{I,X,Y,Z\}$. It is easy to check that applying a uniformly random Pauli operator to any single-qubit density operator results in the maximally mixed state:
$$
\frac{1}{4}\left(\rho + X\rho X + Y\rho Y + Z\rho Z \right) = \frac{\one_1}{2}
\qquad
\text{for all }
\rho \in \states (\H_1)\,.
$$

Since the Pauli operators are self-adjoint, we may implement the above map by choosing two bits $s$ and $t$ uniformly at random and then applying
$$
\rho \mapsto X^s Z^t \rho Z^t X^s\,.
$$
To observers with no knowledge of $s$ and $t$, the resulting state is information-theoretically indistinguishable from $\one_1/2$. Of course, if we know $s$ and $t$, we can invert the above map and recover $\rho$ completely.

The above map can be straightforwardly extended to the $n$-qubit case in order to obtain an elementary {\em quantum encryption scheme} called the {\em quantum one-time pad}.
We first set $X_j = \one^{\otimes j-1} \otimes X \otimes \one^{\otimes n-j}$ and likewise for $Y_j$ and $Z_j$. We define the $n$-qubit Pauli group $\mathcal P_n$ to be the subgroup of $\operatorname{SU}(\H_n)$ generated by $\{X_j, Y_j, Z_j : j = 1, \dots, n \}$. Note that Hermiticity is inherited from the single-qubit case, i.e.\ $P^\dag = P$ for every $P \in \mathcal{P}_n$.

\begin{definition}
[quantum one-time pad]
For $r \in \bit{2n}$, we define the {\em quantum one-time pad (QOTP)} on $n$ qubits with classical key $r$ to be the map:
$$
P_r := \prod_{j=1}^{n} X_{j}^{r_{2j-1}} Z_{j}^{r_{2j}} \in \P_n\,.
$$
\end{definition}
The effect of $P_r$ on any quantum state $\rho \in \states(\H_n)$ is simply
$$
\frac{1}{2^{2n}}\sum_{r \in \bit{2n}} P_r \rho P_r = \frac{\one_n}{2^n}\,.
$$
As before, the map $\rho \mapsto P_r \rho P_r$ (for uniformly random key $r$) is an information-theoretically secure symmetric-key encryption scheme for quantum states.

Just as in the classical case~\cite{Shannon1949}, any reduction in key length is not possible without compromising information-theoretic security~\cite{AMTW00,BR2003}. Of course, in practice the key length of the one-time pad (quantumly or classically) is highly impractical. This is a crucial reason to consider---as we do in this work---encryption schemes which are secure only against computationally bounded adversaries.

\subsection{Efficient Classical and Quantum Computations}\label{sec:efficient}

We will refer to several different notions of efficient algorithms. The most basic of these is a deterministic polynomial-time algorithm (or PT). A PT $\alg A$ is defined by a polynomial-time uniform\footnote{Recall that polynomial-time uniformity means that there exists a polynomial-time Turing machine which, on input~$n$ in unary, prints a description of the $n$th circuit in the family.} family $\alg A := \{\alg A_n\}_{n \in \NN}$ of classical Boolean circuits over some gate set, with one circuit for each possible input size. For a bitstring $x$, we define $\alg A(x) := \alg A_{|x|}(x)$. We say that a function family $f:\bit{n} \rightarrow \bit{m}$ is PT-computable if there exists a PT $\alg A$ such that $\alg A(x) = f(x)$ for all $x$; it is implicit that $m$ is a function of $n$ which is bounded by some polynomial, e.g., the same one that bounds the running time of~$\alg A$.

A probabilistic polynomial-time algorithm (or PPT) is again a polynomial-time uniform family of classical Boolean circuits, one for each possible input size~$n$. The $n$th circuit still accepts $n$ bits of input, but now also has an additional ``coins'' register of $p(n)$ input wires. Note that uniformity enforces that the function $p$ is bounded by some polynomial. For a PPT $\alg A$, $n$-bit input $x$ and $p(n)$-bit coin string $r$, we set $\alg A(x; r) := \alg A_n(x; r)$. In contrast with the PT case, the notation $\alg A(x)$ will now refer to the random variable $\alg A(x; r)$ where $r \inrand \bit{p(n)}$. Overloading notation slightly, $\alg A(x)$ can also mean the corresponding probability distribution; for example, the set of all possible outputs of $\alg A$ on the input $1^n$ is denoted $\supp \alg A(1^n)$.

We define a quantum polynomial-time algorithm (or QPT) to be a polynomial-time uniform family of quantum circuits, each composed of gates that may perform general admissible operations, chosen from some finite, universal set. A commonly-used alternative is to specify that the elements of the gate set are unitary. In terms of computational power, the models are the same~\cite{AKN1998}, however using admissible operations (versus unitary ones only) allows us to formalize a wider range of oracle-enabled QPT machines (see Section~\ref{sec:oracles}).
In general, a QPT $\alg A$ defines a family of admissible maps from input registers to output registers: $\alg A : \states(\H_n) \rightarrow \states\bigl(\H_u\bigr)$. As before, the $n$th circuit in the family will be denoted by $\alg A_n$.
When $\rho$ is an $n$-qubit state, $\alg A(\rho)$  denotes the corresponding $u(n)$-qubit output state (by uniformity, $u$ is bounded by some polynomial). Overloading the notation even further, for $n$-bit strings $x$ we set $\alg A(x) := \alg A(\outerprod{x}{x})$. The expression $\alg A(x) = y$ for classical~$y$ is taken to evaluate to true if the output register of the circuit contains the state $\outerprod{y}{y}$ exactly. Unless explicitly stated, any statements about the probability of an event involving a QPT are taken over the measurements of the QPT, in addition to any indicated random variables. For instance, the expression $\pr_{x \in_R \bit{n}} [\alg A(x) = y]$ means the probability that, given a uniformly random input string $x$, the output register of the $n$th circuit of the QPT $\alg A$ executed on $\outerprod{x}{x}$, after all gates and measurements have been applied, is in the state~$\outerprod{y}{y}$.

At times, we will define QPTs with many input and output quantum registers. In these cases, some straightforward bookkeeping (e.g., via an additional classical register) may be required; for the sake of clarity, we will simply assume that this has been handled.

Throughout this work, we are concerned only with polynomial-time \emph{uniform} computation. That is to say, the circuit families that describe any PT, PPT, or QPT will always be both of polynomial length \emph{and} generatable by some fixed (classical) polynomial-time Turing machine. In particular, we consider uniform adversaries only---although all of our results carry over appropriately to the non-uniform setting as well.

\subsection{Oracles}
\label{sec:oracles}

We denote by $\mathcal A^f$ an algorithm which has oracle access to some function family~$f$. Such an algorithm (whether PT, PPT, or QPT) is defined as above, except each circuit in the algorithm can make use of additional ``oracle gates'' (one for each possible input size) which evaluate~$f$. In the case of PTs and PPTs, oracles can implement any function from bitstrings to bitstrings. In the case of QPTs, we  consider two different oracle types.

First, we allow purely classical oracles. Just as in the case of PTs and PPTs, a classical oracle implements a function $f$ from bitstrings to bitstrings.  In the case of a QPT with a classical oracle, \emph{queries can be made on classical inputs only} (this is sometimes referred to as ``standard-security''~\cite{Zhandry2012}).  We emphasize that we do not require that the oracle is made reversible, nor do we allow the QPT to input superpositions. Note that any such oracle can be implemented by an admissible map, such that classical inputs $x$ are deterministically mapped to $f(x)$ (to see this, start with a Boolean circuit for $f$, make it reversible, and then recall that adding ancillas and discarding output bits are admissible operations). While it might seem that disallowing superposition inputs is an artificial and unrealistic restriction, in our case it actually strengthens results. For instance, we will show that secure quantum encryption can be achieved using pseudorandom functions which are secure only against quantum adversaries possessing just classical oracle access. One can of course also ask for \emph{more powerful} functions (which are secure against superposition access, or ``quantum-secure''~\cite{Zhandry2012}) but this turns out to be unnecessary in our case.
Second, we also allow oracles that are admissible maps. More precisely, for an admissible map family $\C$, we write $\alg A^\C$ to denote a QPT whose circuits can make use of special ``oracle gates'' which implement admissible maps from the family $\C$. Each such gate accepts a  quantum register as input, to which it applies the appropriate admissible map from the family, and returns an output register. It is not necessary for the input and output registers to have the same number of qubits.

In any case, each use of an oracle gate counts towards the circuit length, and hence also towards the total computation time of the algorithm. In particular, no PT, PPT or QPT algorithm may make more than a polynomial number of oracle calls.


\section{Quantum Encryption and Indistinguishability}\label{sec:encryption}

In this section, we give general definitions of encryption schemes for quantum data (\expref{Section}{sec:encryp-schemes}) and a corresponding notion of indistinguishability, including IND-CPA and IND-CCA1 (\expref{Section}{sec:IND-CPA}.)

\subsection{Quantum Encryption Schemes}
\label{sec:encryp-schemes}

We start by defining {\em secret-key encryption for quantum data}. In the following we assume that the secret key is a classical bitstring, while the plaintext and the ciphertext can be arbitrary quantum states. We refer to $\K$, $\H_M$ and $\H_C$ as the key space, the message (or plaintext) space, and the ciphertext space, respectively. We remark that these are actually infinite families of spaces, each with a number of (qu)bits which scales polynomially with~$n$. We assume that $\K := \bit{n}$, so that the key-length is~$n$ bits, and the  plaintext and the ciphertext lengths are $m\leq\poly(n)$ and $c\leq\poly(n)$ qubits, respectively. The key-generation algorithm accepts a description of the security parameter~$n$ in unary and outputs a classical key of length~$n$. Later,  we will define an additional Hilbert space $\H_E$ in order to model auxiliary information used by some adversary. Encryption accepts a classical key and a plaintext, and outputs a ciphertext; decryption accepts a classical key and a ciphertext, and outputs a plaintext. The correctness guarantee is that plaintexts are preserved (up to negligible error) under encryption followed by decryption under the same key.

\begin{definition}\label{def:SKE}
A {\em quantum symmetric-key encryption scheme (or qSKE)} is a triple of QPTs:
\begin{enumerate}
\item (key generation) $\KeyGen: 1^n \mapsto k \in \K$
\item (encryption) $\Enc: \K \times \states(\H_M) \rightarrow \states(\H_C)$
\item (decryption) $\Dec: \K \times \states(\H_C) \rightarrow \states(\H_M)$
\end{enumerate}
such that $\| \Dec_k \circ \Enc_k - \one_M \|_\diamond \leq \negl(n)$
for all $k \in \emph{\supp} \KeyGen(1^n)$.
\end{definition}

In the above, we used a convenient shorthand notation for encryption and decryption maps with a fixed key $k$ (which is classical), formally defined by $\Enc_k : \rho \mapsto \Enc(k, \rho)$ and
$\Dec_k : \sigma \mapsto \Dec(k, \sigma).$

Next, we define a notion of {\em public-key encryption for quantum data}. In addition to the usual spaces from the symmetric-key setting above, we now also have a public key of length $p(n) \leq $ poly$(n)$ bits. We define the related public-key space as $\K_{pub} \subset \bit{p}$ and reuse $\K$ for the corresponding private-key space.

\begin{definition}\label{def:PKE}
A {\em quantum public-key encryption scheme (or qPKE)} is a triple of QPTs:
\begin{enumerate}
\item (key-pair generation) $\KeyGen: 1^n \mapsto (pk,sk) \in \K_{pub} \times \K$
\item (encryption with public key) $\Enc: \K_{pub} \times \states(\H_M) \rightarrow \states(\H_C)$
\item (decryption with private key) $\Dec: \K \times \states(\H_C) \rightarrow \states(\H_M)$
\end{enumerate}
such that $\| \Dec_{sk} \circ \Enc_{pk} - \one_m \|_\diamond \leq \negl(n)$
for all $(pk, sk) \in \emph{\supp} \KeyGen(1^n)$.
\end{definition}
In this case, we again placed the relevant keys in the subscript, i.e.,
\begin{equation*}\label{eq:pub_enc_dec}
\Enc_{pk} : \rho \mapsto \Enc(pk, \rho)
\qquad \text{and} \qquad
\Dec_{sk} : \sigma \mapsto \Dec(sk, \sigma)\,.
\end{equation*}

We remark that some variations of the above two definitions are possible. For instance, one could demand that encryption followed by decryption is exactly equal to the identity operator. The schemes we present in \expref{Section}{sec:Constructions} will in fact satisfy this stronger condition. 

\subsection{Indistinguishability of Encryptions}
\label{sec:IND-CPA}

Following the classical definition, the security notion of \emph{quantum indistinguishability under chosen plaintext attacks} has been considered previously for the case of quantum
encryption schemes 
in~\cite{BJ15} and for
classical encryption schemes
in~\cite{Gagliardoni2015}. Here, we present the definition from~\cite{BJ15}, which we slightly extend to the  CCA1 (chosen ciphertext
attack) setting.
The security definitions are formulated with the public-key (or
asymmetric-key) setting in mind, and we clarify  when meaningful
differences in the symmetric-key setting arise.

Our definition models a situation in which
an honest user encrypts
messages of the adversary's choice; the adversary then attempts to match the ciphertexts to the plaintexts. In our formulation, an IND adversary consists of two QPTs: the {\em message generator} and the {\em distinguisher}. The message generator takes as input the security parameter and a public key, and outputs a challenge state consisting of a plaintext and some auxiliary information. The auxiliary information models, for instance, the fact that the output state might be entangled with some internal state of the adversary itself. Then the distinguisher receives this auxiliary information, and a state which might be either the encryption of the original challenge state or the encryption of the zero state. The distinguisher's goal is to decide which of the two is the case.

Security in this model requires that the adversary does not succeed with probability significantly better than guessing. We also define two standard variants: indistinguishability under chosen plaintext attack (IND-CPA) and indistinguishability under chosen-ciphertext-attack (IND-CCA1). We leave the definition of CCA2 (adaptive chosen ciphertext attack) security as an interesting open problem. As before, all circuits are indexed by the security parameter.

\begin{definition}[IND]\label{def:IND}
A qPKE scheme $(\KeyGen, \Enc, \Dec)$ has {\em indistinguishable encryptions} (or is {\em IND secure}) if for every QPT adversary $\A=(\M,\D)$ we have:
\begin{equation*}
\left| \Pr \left[ \; \D\big\{ (\Enc_{pk} \otimes \one_E) \rho_{ME} \big\} = 1 \; \right] -
\Pr \left[ \; \D\big\{ (\Enc_{pk} \otimes \one_E) (\egoketbra{0}_M \otimes \rho_E) \big\} = 1 \;  \right] \right| \leq \negl(n)
\end{equation*}
where $\rho_{ME} \from \M(pk)$, $\rho_E = \tr_M(\rho_{ME})$, and the probabilities are taken over $(pk, sk) \leftarrow \KeyGen(1^n)$ and the internal randomness of \Enc, $\M$, and $\D$.
\begin{itemize}
\item \textbf{\emph{IND-CPA:}} In addition to the above, $\M$ and $\D$ are given oracle access to $\Enc_{pk}$.
\item \textbf{\emph{IND-CCA1:}} In addition to IND-CPA, $\M$ is given oracle access to $\Dec_{sk}$.
\end{itemize}
\end{definition}

Here we use $\egoketbra{0}_M$ to denote $\egoketbra{0^m}$, where $m$ is the number of qubits in the $M$ register.

The definition is illustrated in \expref{Figure}{fig:ind}. The symmetric-key scenario is the same, except $pk = sk$, and $\M$ receives only a blank input. We remark that in the public-key setting, IND implies IND-CPA: an adversary with knowledge of $pk$ can easily simulate the $\Enc_{pk}$ oracle. Note that, under CPA, the IND definition is known to be equivalent to IND in the \emph{multiple-message} scenario~\cite{BJ15}.

\begin{figure}[h]
\begin{center}
\begin{tikzpicture}[auto, node distance=7em, >=latex']
\node at (-0.2,0) [input](input){{\scriptsize $pk$}};
\node at (1,0) [block, minimum height=3em](M){$\M$};
\node at (5.5, 0.31) [block](Enc){$\Enc_{pk}$};
\node at (7.5, 0) [block, minimum height=3em](D){$\D$};
\node at (8.5, 0) [output](output){};
\draw
 (input) edge[double] node {} (M)
 (M.30) edge node [pos=0.14] {$M$} (Enc)
 (M.330) edge node [pos=0.09] {$E$} (D.210)
 (Enc) edge node {} (D.150)
 (D) edge[double] node {} (output);
\end{tikzpicture}

\bigskip

\begin{tikzpicture}[auto, node distance=7em, >=latex']
\node at (-0.2,0) [input](input){{\scriptsize $pk$}};
\node at (1,0) [block, minimum height=3em](M){$\M$};
\node at (2.65, 0.31) [trace](trace){};
\node at (4, 0.31) [input](0){{\scriptsize $|0 \rangle$}};
\node at (5.5, 0.31) [block](Enc){$\Enc_{pk}$};
\node at (7.5, 0) [block, minimum height=3em](D){$\D$};
\node at (8.5, 0) [output](output){};
\draw
 (input) edge[double] node {} (M)
 (M.30) edge node [pos=0.5] {$M$} (trace)
 (0) edge[double] node [pos=0.5] {} (Enc)
 (M.330) edge node [pos=0.09] {$E$} (D.210)
 (Enc) edge node {} (D.150)
 (D) edge[double] node {} (output);
\end{tikzpicture}
\caption{IND posits that a QPT $(\M, \D)$ cannot distinguish between these two scenarios.} \label{fig:ind}
\end{center}
\end{figure}
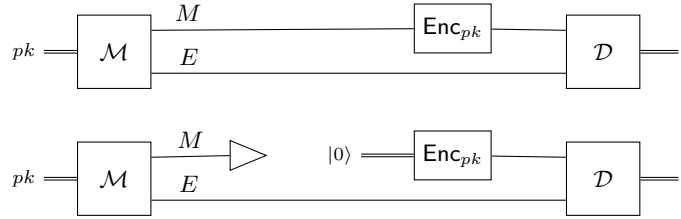



\section{Quantum Semantic Security}
\label{sec:quantum-sem}

This section is devoted to defining quantum semantic security~(\expref{Section}{sec:defining-semantic}), and showing its equivalence with quantum indistinguishability~(\expref{Section}{sec:IND-equiv-SEM}).

Following the classical definition, the security notion of \emph{quantum semantic security under chosen plaintext attacks} has been  
given previously in~\cite{Gagliardoni2015}
for the case of
a special class of quantum states arising when considering quantum access to classical encryption schemes.
Here, we give a more general definition for arbitrary quantum plaintexts. As we outlined the classical situation with semantic security in \expref{Section}{sec:intro-SEM}, we start with a discussion of some difficulties in transitioning to the quantum setting. A similar discussion can be found in~\cite{Gagliardoni2015} and we explain below where and why we make different choices.

\subsection{Difficulties in the Quantum Setting}
When attempting to transfer the definition of semantic security to the quantum world, the main question one encounters is to determine the quantum equivalents of $h(x)$ and $f(x)$ (as it is relatively clear that the plaintext $x$ would have as quantum equivalent a quantum state $\rho_{M}$, in a \emph{message register}, $M$).

For the case of the side-information, $h(x)$,  one might attempt to postulate that this side information is available via the output of a quantum map $\Phi_h$, evaluated on~$\rho_{M}$. There are, however, two obvious problems with this approach: firstly, it is unclear how to \emph{simultaneously} generate both $\rho_{M}$ and $\Phi_h(\rho_{M})$ (the main obstacle stemming from the quantum \emph{no-cloning} theorem~\cite{Wootters1982}, according to which it is not possible to perfectly copy an unknown quantum state)\footnote{\cite{Gagliardoni2015} solves the issue by requiring a quantum circuit that takes classical randomness as input and outputs plaintext states. Hence, multiple plaintext states can be generated by using the same randomness.}. Secondly, it is well-established that the most general type of quantum side-information includes entanglement (contrary to the scenario studied in~\cite{Gagliardoni2015}). We therefore conclude that side information should be modelled simply as an extra register (called~$E$) such that $\rho_{ME}$ are in an arbitrary quantum state (as generated by some process---for a formal description, see \expref{Definition}{def:SEM}).

For the case of the target function $f$, one might also postulate a quantum map~$\Phi_f$, the goal then (for both the adversary and simulator), being to output $\Phi_f(\rho_M)$. However, given that quantum states and maps form a continuum, one must exercise care in quantifying when a simulator has successfully simulated the adversary.
We propose three possible tests for quantifying ``success'' in the semantic security game, each leading to its own definition.
Since we show  that all three definitions are equivalent, we conclude that it is a matter of taste (or context) which definition to label as \emph{the} definition of quantum semantic security. 
We focus in this section on the first one, which we called SEM, because we find that it the most natural.
We give formal definitions and proofs of equivalence for all three definitions in Appendix~\ref{app:alternate-Sem-Sec}. 
Here is an overview of the three different notions:
\begin{itemize}
\item \textbf{SEM.} In Definition~\ref{def:SEM}, a state  $\rho_{MEF}$ is generated; intuitively, the contents of register $\reg F$ can be seen as a ``target'' output that the adversary tries to achieve (however, this is not quite the case as we point out shortly). We then postulate a quantum polynomial time \emph{distinguisher} who is given the $F$ register and charged with distinguishing the output of the adversary from the output of the simulator, with security being associated with the inability of the distinguisher in telling the two situations apart. We thus see that the role of register~$\reg F$ is actually to assist the distinguisher: semantic security corresponds to the situation where the distinguisher essentially cannot tell the real from ideal apart, \emph{even with access to the $\regF$ system}.
\item \textbf{SEM2}. In \expref{Definition}{def:SEMtwo}, we specify instead that the state $\rho_{MEF}$ be a \emph{classical-quantum state}. That is, $\rho_{ME}$ is quantum, but the register $\reg F$ contains a classical state. Thus, correlations shared between the two systems are classical only. The requirement for security is that the simulator should provide a classical output that equals the contents of $\reg F$, essentially just as well as the adversary can.
\item  \textbf{SEM3}. In \expref{Definition}{def:SEMthree}, we introduce a classical function~$f$, thus closely mimicking the classical definition. Namely, we specify as in SEM2 that $\reg F$ contains a classical state~$y$, which we furthermore assume to be precisely the results of any measurements used to generate~$\rho_{ME}$ (thus, $y$ is, in a sense, a full ``classical description'' of~$\rho_{ME}$). The requirement for security is that the simulator is able to output $f(y)$ (for any~$f$) with essentially the same probability as the adversary.
\end{itemize}

\subsection{Definition of Semantic Security}
 \label{sec:defining-semantic}

As before, we work primarily in the public-key setting; adaptation to the symmetric-key setting is again straightforward.
In our concrete formulation of~SEM (\expref{Definition}{def:SEM}), we define the following QPT machines: the \emph{message generator} $\M$ (which generates $\rho_{MEF}$), the  \emph{adversary} $\A$, the \emph{simulator} $\S$ and the \emph{distinguisher}~$\D$.
\begin{definition}\label{def:SEM} [SEM]
A qPKE scheme $(\KeyGen, \Enc, \Dec)$ is \emph{semantically secure} if for any QPT adversary $\A$, there exists a QPT simulator $\S$ such that for all QPTs $\M$ and $\D$,
\begin{equation*}
\left|\Pr \left[ \; \D\big\{ (\A \otimes \one_F) (\Enc_{pk} \otimes \one_{EF})\rho_{MEF} \big\} = 1 \; \right] -
\Pr \left[ \; \D \big\{ (\S \otimes \one_F)\rho_{EF} \big\} = 1 \; \right] \right|
\leq \negl(n)\,,
\end{equation*}
where $\rho_{MEF} \from \M(pk)$, $\rho_{EF} = \tr_M(\rho_{MEF})$, and the probability is taken over $(pk, sk) \leftarrow \KeyGen(1^n)$ and the internal randomness of \Enc, $\A$, $\S$ and $\D$.
\begin{itemize}
\item \textbf{\emph{SEM-CPA:}} In addition to the above, all QPTs are given oracle access to $\Enc_{pk}$.
\item \textbf{\emph{SEM-CCA1:}} In addition to IND-CPA, $\M$ is given oracle access to $\Dec_{sk}$.
\end{itemize}
\end{definition}

The interactions among the QPTs are illustrated in \expref{Figure}{fig:SEM}. A few remarks are in order. First, all the registers above are uniformly of size polynomial in~$n$. Second, the input and output registers of the relevant QPTs are understood from context, e.g., the expression $(\S \otimes \one_F) \rho_{EF}$ makes clear that the input register of~$\S$ is~$E$. Third, we note that SEM implies SEM-CPA in the public-key setting, since access to the public key  implies simulatability of~$\Enc_{pk}$. Finally, just as in the case of IND, adapting to the symmetric-key setting is simply a matter of setting $pk = sk$ and positing that $\M$ receives only a blank input.

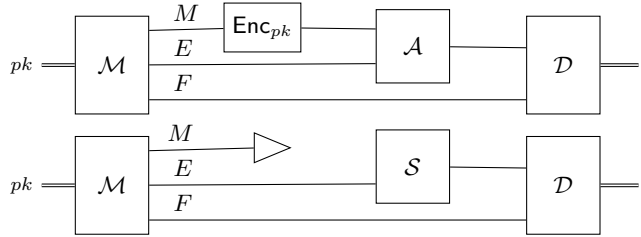
\begin{figure}
\begin{center}
\begin{tikzpicture}[auto, node distance=7em, >=latex']
\node at (-0.2,0) [input](input){{\scriptsize $pk$}};
\node at (1,0) [block, minimum height=4em](M){$\M$};
\node at (3, 0.5) [block](Enc){$\Enc_{pk}$};
\node at (5, 0.25) [block, minimum height=3em](A){$\A$};
\node at (7, 0) [block, minimum height=4em](D){$\D$};
\node at (8, 0) [output](output){};
\draw (input) edge[double] node {} (M)
 (M.43) edge node [pos=0.5] {$M$} (Enc)
 (M) edge node [pos=0.15] {$E$} (A.205)
 (M.317) edge node [pos=0.09] {$F$} (D.223)
 (Enc) edge node {} (A.155)
 (A) edge node {} (D.155)
 (D) edge[double] node {} (output);
\end{tikzpicture}

\begin{tikzpicture}[auto, node distance=7em, >=latex']
\node at (-0.2,0) [input](input){{\scriptsize $pk$}};
\node at (1,0) [block, minimum height=4em](M){$\M$};
\node at (3, 0.5) [trace](trace){};
\node at (5, 0.25) [block, minimum height=3em](A'){$\S$};
\node at (7, 0) [block, minimum height=4em](D){$\D$};
\node at (8, 0) [output](output){};
\draw (input) edge[double] node {} (M)
 (M.43) edge node [pos=0.3] {$M$} (trace)
 (M) edge node [pos=0.15] {$E$} (A'.205)
 (M.317) edge node [pos=0.09] {$F$} (D.223)
 (A') edge node {} (D.155)
 (D) edge[double] node {} (output);
\end{tikzpicture}
\caption{SEM: for all adversaries $\A$ there exists a simulator $\S$ such that these two scenarios are indistinguishable.}
\label{fig:SEM}
\end{center}
\end{figure}

The classical (uniform) definition of semantic security is recovered as a special case, as follows. All of the QPTs are PPTs, and the message generator $\M$ outputs classical plaintext $m$, side information $h(m)$ and target function $f(m)$. The distinguisher $\D$ simply checks whether the adversary's (or simulator's) output is equal to the contents of the $F$ register.

\subsection{Semantic Security is Equivalent to Indistinguishability}
\label{sec:IND-equiv-SEM}

While semantic security gives a strong and intuitively meaningful definition of security, indistinguishability is typically easier to prove and work with. In this section we show that---just as in the classical setting---the two notions are equivalent. This proves \expref{Theorem}{thm:intro:ind-equiv-sem}.
 The equivalence holds for all of the variants of \expref{Definition}{def:IND} and \expref{Definition}{def:SEM}: under either public or private-key, we have equivalence of IND with SEM, IND-CPA with SEM-CPA, and IND-CCA1 with SEM-CCA1. Here, we focus on the SEM definition; see \expref{Appendix}{app:alternate-Sem-Sec} for the equivalence with the SEM2 and SEM3 definitions.

\begin{theorem}[IND $\implies$ SEM]
\label{thm:ind''-implies-sem}
If a quantum encryption scheme \\* $(\KeyGen, \Enc, \Dec)$ has indistinguishable encryptions (IND), then it is semantically secure (SEM).
\end{theorem}

\begin{proof}
Suppose that an encryption scheme $(\KeyGen, \Enc, \Dec)$ has indistinguishable encryptions.  Let $\A$ be QPT SEM attacker against semantic security as in \expref{Definition}{def:SEM}. We define the QPT SEM simulator $\S$ as follows: $\S$ does not receive $\Enc_{pk}(\rho_M)$, but instead runs $\A$
 on input $(\Enc_{pk} \otimes \one_E) (\egoketbra{0} \otimes \rho_E)$ and outputs whatever $\A$ outputs. Let $\M$ be a QPT SEM message generator that outputs~$\rho_{MEF}$.

Assume for a contradiction the existence of a QPT SEM distinguisher $\D$ which successfully distinguishes the output of $\A$ from the output of $\S$ (with the help of register $F$), then the combination of $\A$ and $\D$
successfully distinguishes $(\Enc_{pk} \otimes I_{EF})\rho_{MEF}$ from $(\Enc_{pk} \otimes I_{EF}) (\egoketbra{0} \otimes \rho_{EF})$, hence contradicting the indistinguishability. \qed

\end{proof}

In the private-key setting without CPA oracle access, 
$\S$ runs $\KeyGen(1^n)$ to generate his own secret key $k'$, and then encrypts $\egoketbra{0^n}$ using $k'$ instead of $k$. The ciphertexts $\Enc_{k} \egoketbra{0}$ and $\Enc_{k'} \egoketbra{0}$ will be distributed identically since $k$ and $k'$ are. Hence, the success probability of the SEM simulator $\S$ does not change.

In case of CPA and CCA1 oracles, both for the public- and private-key setting, the simulator $\mathcal{S}$ forwards $\mathcal{A}$'s oracle queries to his own oracle(s), and $\mathcal{S}$ obtains $\mathcal{A}$'s input state by a call to his encryption oracle on state $\egoketbra{0}$, joined with his auxiliary information $\rho_E$.

\begin{theorem}[SEM $\implies$ IND]
\label{thm:sem-implies-ind''}
If a quantum encryption scheme \\* $(\KeyGen, \Enc, \Dec)$ is semantically secure (SEM),  then it has indistinguishable encryptions (IND).
\end{theorem}
\begin{proof}
Let $(\M,\D)$ be an IND adversary such that $\D$ distinguishes $(\Enc_{pk} \otimes \one_{E})\rho_{ME}$ from $(\Enc_{pk} \otimes \one_{E}) (\egoketbra{0} \otimes \rho_{E})$ with advantage $\eps(n)$ if  $\rho_{ME} \from \M$. Let us consider the SEM message generator $\M'$ which runs $\rho_{ME} \from \M$ and outputs (with probability $\frac12$ each) either the state $\rho_{ME} \otimes \egoketbra{0}_F$ or the state $\egoketbra{0}_M \otimes \rho_{E} \otimes \egoketbra{1}_F$. Next we consider the SEM attacker $\A$ which runs $\D$ and outputs the classical bit that $\D$ outputs. We also consider the SEM attacker $\A \oplus 1$, which outputs the opposite bit. As SEM distinguisher, let us consider the procedure which compares $\A$'s output bit to a measurement (in the computational basis) of the qubit in register $F$. Any SEM simulator $\mathcal{S}$ that does not have access to the encrypted $M$-register has to guess the state of the random bit in $F$ and will be correct with probability $1/2$. Then $\eps(n)$ is twice the maximum of the advantages that $\A$ and $\A \oplus 1$ have in successfully predicting $F$ over $1/2$, the probability of success of any simulator. By SEM, both of these advantages are negligible, and hence so is $\eps(n)$.\qed
\end{proof}



\section{Quantum Encryption Schemes}
\label{sec:Constructions}

We now turn to the question of existence for encryption schemes for quantum data. We present two schemes based on the existence of classical functions which are difficult to invert for quantum computers. The first scheme (\expref{Section}{sec:construction-SKE}) is symmetric-key and IND-CCA1-secure; the second scheme (\expref{Section}{sec:CPA-public-key}) is public-key and IND-CPA-secure. By the  results of Section~\ref{sec:quantum-sem}, these schemes are also semantically secure.

\subsection{Quantum Symmetric-Key Encryption from One-Way Functions}
\label{sec:construction-SKE}

In this section, we prove \expref{Theorem}{thm:intro:SKE}: \emph{\theoremStatementqOWF}

The proof proceeds in two steps. First, we define quantum-secure one-way functions (qOWFs) and  quantum-secure pseudo-random functions (qPRFs); we can argue as in the classical world that qPRFs exist if qOWFs do (\expref{Theorem}{thm:qOWFsimpliesqPRFs}.) 
Second, we show that any qPRF can be used to construct an explicit IND-CCA1-secure symmetric-key scheme for quantum data.

\medskip
We begin with the formal definitions of qOWFs and qPRFs, and a statement of the result connecting the two.

\begin{definition}\label{def:qOWF}
A PT-computable function $f : \bit{*} \rightarrow \bit{*}$ is a \emph{quantum-secure one-way function (qOWF)} if for every QPT $\mathcal A$,
\begin{equation*}
\underset{x \inrand \bit{n}}{\pr}[\mathcal A(f(x), 1^n) \in f^{-1}(f(x))] \leq \negl(n)\,.
\end{equation*}
\end{definition}

\begin{definition}\label{def:qPRF}
A PT-computable function family $f : \bit{n} \times \bit{m} \rightarrow \bit{\ell}$ is a \emph{quantum-secure pseudorandom function (qPRF)} if for every QPT $\mathcal D$ equipped with a classical oracle,
\begin{equation*}
\left|\underset{k \inrand \bit{n}}{\pr}[\mathcal D^{f_k}(1^n) = 1]~~~ - \underset{g \inrand \{\bit{m} \rightarrow \bit{\ell}\}}{\pr}[\mathcal D^g(1^n) = 1] \right | \leq \negl(n)\,.
\end{equation*}
\end{definition}

\noindent We remark that, to some readers, the restriction to classical oracles might seem artificial. While one can certainly consider functions with the \emph{stronger} guarantee of resistance to quantum adversaries with quantum oracle access, stronger functions are not necessary to establish our results. We thus opt for the weaker primitive. In either case, the following holds.

\begin{theorem}
\label{thm:qOWFsimpliesqPRFs}
If qOWFs exist, then qPRFs exist.
\end{theorem}
Since our definitions are in terms of \emph{classical} oracles, the classical proof that shows that qOWFs imply qPRFs carries through~\cite{HAstad:1999:PGA:312173.312213,Goldreich1986} . We remark that Zhandry~\cite{Zhandry2012} extended this result to the case of functions secure against quantum superposition queries, what he calls ``quantum-secure PRFs.''
It should be noted that the proof of the Theorem~\ref{thm:qOWFsimpliesqPRFs} actually implies the existence of a qPRF for any (polynomial) choice of the parameters $m$ and $\ell$ in \expref{Definition}{def:qPRF}. 

We are now ready to proceed with the second part of the proof of \expref{Theorem}{thm:intro:SKE}, namely the construction of an encryption scheme from a given qPRF. Essentially, this scheme encrypts a quantum state~$\rho$ by first selecting a random string $r$, then inputing~$r$ into a qPRF; the output $f_{k(r)}$ is then used as an encryption key for the quantum one-time pad, $P_{f_{k(r)}}$.

\begin{scheme}\label{defi:SKE-from-qPRF}
Let $f : \bit{n} \times \bit{2n} \rightarrow \bit{2n}$ be a qPRF. Let qPRF-SKE be the following triple of QPT algorithms:
\begin{enumerate}
\item (key generation) $\KeyGen(1^n)$: output $k \inrand \bit{n}$;
\item (encryption) $\Enc_k(\rho)$: choose $r \inrand \bit{2n}$ and output $\ket{r}\bra{r} \otimes P_{f_k(r)}\rho P_{f_k(r)}$.
\item (decryption) $\Dec_k(\sigma):$ measure the first $2n$ qubits in the computational basis to obtain $r' \in \bit{2n}$; apply $P_{f_k(r')}$ to remaining $2n$ qubits and output the result.
\end{enumerate}
\end{scheme}
For simplicity, we chose $\states(\H_n)$ for the key space and the plaintext space, and $\states(\H_{2n})$ for the ciphertext space; we can easily adapt the above to other polynomially-related cases by selecting a qPRF with different parameters.
Correctness of Scheme~\ref{defi:SKE-from-qPRF} is easily verified:
\begin{equation*}
\Dec_k (\Enc_k (\rho))
= \Dec_k \bigr(\egoketbra{r} \otimes P_{f_k(r)}\rho P_{f_k(r)}\bigl)
= P_{f_k(r)} P_{f_k(r)}\rho P_{f_k(r)} P_{f_k(r)}
= \rho\,,
\end{equation*}
where the second equality follows from the definition of the decryption function and the last step is due to the fact that the Pauli operators are self-inverse. Next, we show that the scheme is secure against non-adaptive chosen ciphertext attacks. The classical version of this result is standard, and we use essentially the same proof; see, e.g., Proposition 5.4.18 in Goldreich's textbook~\cite{Goldreich}.

\begin{lemma}
\label{lem:qPRF-implies-SKE}
If $f$ is a qPRF, then \expref{Scheme}{defi:SKE-from-qPRF} is an IND-CCA1-secure symmetric-key quantum encryption scheme as defined in \expref{Definition}{def:IND}.
\end{lemma}

\begin{proof}
First, we analyse the security of the scheme in an idealized scenario where $f$ is a truly random function. 
We claim that in this case, $\alg A$ correctly guesses the challenge with probability at most $1/2 + \negl(n)$
(see \expref{Definition}{def:IND'}).
 In fact, this bound holds for a stronger adversary $\alg A'$, who has access to a classical oracle for~$f$ prior to the challenge, and access to polynomially-many pairs $(r_i, f(r_i))$ for random $r_i, 1 \leq i \leq q$, after the challenge. This adversary is stronger than $\alg A$ since it can simulate~$\alg A$ by implementing $\Enc_f$ and $\Dec_f$ oracles using its $f$ oracles. Since the input $r$ into $f$ in the challenge ciphertext is uniformly random, the probability that any of the polynomially-many oracle calls of $\alg A'$ uses the same $r$ is negligible. In the case that no oracle calls use $r$, the mixtures of the inputs to $\alg A'$ (including the pairs $(r_i, f(r_i))$) are the same for the original challenge and the zero challenge. This fact can be verified by first averaging over the values of $f(r)$: since $f$ is uniformly random, $f(r)$ is also uniformly random as well as independent of the other values of $f$. In both cases, applying the quantum one-time pad results in the state:
\begin{equation*}
\egoketbra{r} \otimes \dfrac{1}{2^n} \one \otimes \rho_E \otimes \egoketbra{r_1} \otimes \egoketbra{f(r_1)} \otimes \dots \otimes \egoketbra{r_q} \otimes \egoketbra{f(r_q)},
\end{equation*}
and indistinguishability follows.

Next, we consider the case that $f$ is a pseudorandom function. We show that a successful IND-CCA1 adversary $\alg A$ (i.e., one that distinguishes challenges with better than negligible probability) can be used to construct a successful $f$-adversary $\alg A_0$ (i.e., one that distinguishes $f$ from random with non-negligible probability.) The adversary $\alg A_0$ is a QPT with classical oracle access to a function $\varphi : \bit{2n} \rightarrow \bit{2n}$, and aims to output $0$ if $\varphi$ is perfectly random and $1$ if $\varphi = f_k$ for some $k$. Define the simulated oracles
$$
\Enc_\varphi : \rho \mapsto \Bigl(r, \,P_{\varphi(r)} \rho P_{\varphi(r)}\Bigr)
\text{ for }r \inrand \bit{2n}
\quad \text{and} \quad
\Dec_\varphi : \ket{r'}\bra{r'} \otimes \rho \mapsto P_{\varphi(r')} \rho P_{\varphi(r')}\,,
$$
where, as before, we assume that $\Dec_\varphi$ measures the first register before decrypting the second. Note that if $\varphi = f_k$ then these are exactly the encryption and decryption oracles (with key $k$) of the qPRF-SKE scheme.

The QPT $\alg A_0^\varphi$ proceeds as follows. First, it simulates $\alg A$, and replies to its queries to the encryption oracle with $\Enc_\varphi$ and its queries to the decryption oracle with $\Dec_\varphi$. When it transmits the challenge, $\alg A_0^\varphi$ replies with either the encryption of the challenge, or the encryption of $\ket{0^n}\bra{0^n}$, each with probability $1/2$. If $\alg A$ responds correctly, $\alg A_0^\varphi$ outputs $1$; otherwise it outputs $0$. If $\varphi = f_k$ then we have exactly simulated the IND-CCA1 game with adversary $\alg A$; in that case, since $\alg A$ is IND-CCA1-breaking, $\alg A_0^\varphi$ outputs $1$ with probability at least $1/2 + 1/p(n)$ for some polynomial~$p$, for infinitely many $n$.

We conclude that
\begin{equation*}
\left|\underset{k \inrand \bit{n}}{\pr}[\mathcal A_0^{f_k}(1^n) = 1]~~~ - \underset{\varphi \inrand \{\bit{2n} \rightarrow \bit{2n}\}}{\pr}[\mathcal A_0^\varphi(1^n) = 1] \right | \geq 1/p(n) - \negl(n)\,,
\end{equation*}
for infinitely many $n$, i.e., $f$ is not a qPRF.\qed\end{proof}

Putting together \expref{Theorem}{thm:qOWFsimpliesqPRFs} and \expref{Lemma}{lem:qPRF-implies-SKE}, we arrive at a proof of \expref{Theorem}{thm:intro:SKE}.

\subsection{Quantum Public-Key Encryption from Trapdoor Permutations}
\label{sec:CPA-public-key}
For the construction of public-key schemes, we will need qOWFs with an additional property: the existence of \emph{trapdoors} which enable efficient inversion. Following the classical approach of Diffie and Hellman~\cite{hellman1976new}, we set down the notion of a quantum-secure trapdoor one-way permutation (or qTOWP), and then show how to use any qTOWP to construct IND-CPA secure public-key encryption schemes for quantum data. This will establish \expref{Theorem}{thm:intro:PKE}: \emph{\theoremStatementqTOWP}.

We begin with a definition of qTOWPs. We require a slight (but standard) variation of \expref{Definition}{def:qOWF}, namely the notion of a quantum-secure one-way permutation (or qOWP). A qOWP is a qOWF whose input domains are sets $D_i$; moreover, the function restricted to any such domain must be a permutation (from the domain to the corresponding range.) When we augment such a qOWP with trapdoors, we arrive at the following definition.
\begin{definition}\label{def:qTOWP}
A \emph{quantum-secure trapdoor one-way permutation (qTOWP)} is a qOWF
$$
\{f_i: D_i \rightarrow \bit{*}\}_{i \in I}
$$
(where each $f_i$ is a bijection), together with a triple of PPTs $(\alg G, \alg S, \alg I)$ which
\begin{enumerate}
\item (generate (index, trapdoor) pair) $\supp \alg G(1^n) \subseteq (I \cap \bit{n}) \times \bit{n}$;
\item (sample from domain) for all $i \in I$, $\supp \alg S(i) = D_i$;
\item (invert using trapdoor) for all $(i, t) \in \supp \alg G(1^n)$ and all $x \in D_i$, \mbox{$\alg I(f_i(x), t) = x$.}
\end{enumerate}
\end{definition}

Before we can describe the public-key scheme and prove its security, we need two additional (well-known) primitives which can be constructed from any qOWP, with or without trapdoors. The first is a quantum-secure ``hard-core'' predicate, which is a ``yes'' or ``no'' question about inputs $x$ which is difficult to answer if one only knows $f(x)$.

\begin{definition}\label{def:hard-core}
A PT-computable $b:\bit{*} \rightarrow \{0, 1\}$ is a \emph{hard-core} of a qOWP $f$ if for every QPT $\alg A$,
$$
\underset{x \inrand \bit{n}}{\pr} \left[\alg A (f(x), 1^n) = b(x) \right] \leq \frac{1}{2} + \negl(n)\,.
$$
\end{definition}

\begin{theorem}
(\cite{AdcockCleve2002}, quantum analogue of~\cite{GoldreichLevin89}) If qOWPs exist, then qOWPs with hard-cores exist.
\end{theorem}

\noindent The other primitive we need is a quantum-secure pseudorandom generator, which is defined below. The classical proof that hard-cores imply pseudorandom generators carries over with little modification (see \expref{Lemma}{lemma:prg-construction}).

\begin{definition}\label{def:qPRG}
A PT-computable deterministic function $G : \bit{n} \rightarrow \bit{m}$ is a \emph{quantum-secure pseudorandom generator (qPRG)} if for every QPT $\mathcal D$,
$$
\left|\underset{s \inrand \bit{n}}{\pr}[\mathcal D(G(s)) = 1]~~~ - \underset{y \inrand \bit{m}}{\pr}[\mathcal D(y) = 1] \right | \leq \negl(n)\,.
$$
\end{definition}

\begin{lemma}\label{lemma:prg-construction}
Suppose $f$ is a qOWP, $b$ its hard-core predicate, and let $t$ be polynomial in $n$.  Then
$G: s \mapsto b(f^{t-1}(s)) b(f^{t-2}(s)) \dots b(s)$ is a qPRG.
\end{lemma}
\begin{proof}[Sketch] The proof proceeds almost identically as in the classical case (see, e.g., ~\cite{Goldreich2004}.) Let $\mathcal D$ be a quantum adversary that distinguishes $G(U_n)$ from uniform. Note that, as stated in Definition~\ref{def:qPRG}, $\mathcal D$ gets only classical bitstring outputs from the pseudorandom generator. In the classical proof, one constructs an adversary $\mathcal A$ which uses $\mathcal D$ as a black-box subroutine, and breaks the hard-core of $f$. We use the exact same $\mathcal A$ now; in particular, we only need to invoke $\mathcal D$ on classical inputs and read out its (post-measurement) classical outputs ($0$ or $1$). Of course, by virtue of needing to invoke $\mathcal D$, $\mathcal A$ itself will now be a QPT.

In slightly greater detail, we use a standard hybrid argument to give a ``predictor" algorithm $\mathcal A$ that, for some index $i\leq t$, can predict the $i+1$\textsuperscript{st} bit of $G(U_n)$, given as input the first $i$ bits of the output of $G$. $\mathcal A$ succeeds with non-negligible advantage over random, i.e., the probability over $s$ that $\mathcal A(b(f^{t-1}(s))\dots b(f^{t-i}(s)))$ outputs $b(f^{t-(i+1)}(s)$ is at least $1/2+1/p(n)$ where $p(n)$ is some polynomial. Crucially, since $f$ implements a permutation over $\bit{n}$, we have that $b(f^{i-1}(U_n))\dots b(U_n)$ is distributed identically to \\* $b(f^{t-1}(U_n))\dots b(f^{t-i}(U_n))$.  Therefore, given uniform $x$, and $y=f(x)$, we can use the output of the predictor, $A(b(f^{i-1}(y))\dots b(y))=A(b(f^{i}(x))\dots b(f(x)))$ to predict $b(x)$ with non-negligible advantage, in violation of the security guarantee of the hard-core predicate.
\qed\end{proof}

We now have all of the ingredients needed to describe a public-key scheme for encrypting quantum data.
\begin{scheme}\label{def:PKE-from-qTOWF}
Let $f$ be a qTOWP, and let $b$ and $G : \bit{n} \rightarrow \bit{2n}$ be a corresponding hard-core and qPRG, respectively. Let qTOWP-PKE be the following triple of algorithms:
\begin{enumerate}
\item ((public, private) key-pair generation) $\KeyGen(1^n)$: output $\alg G(1^n) = (i, t) \in \bit{n} \times \bit{n}$;
\item (encryption with public key)  $\Enc_i(\rho)$:
\begin{itemize}
\item apply $\alg S(i)$ to select $d \in D_i$, and compute $r := G(d)$;
\item output $\ket{f_i^{2n}(d)}\bra{f_i^{2n}(d)} \otimes P_r \rho P_r$
\end{itemize}
\item (decryption with private key) $\Dec_t(\ket{s}\bra{s} \otimes \sigma):$
\begin{itemize}
\item for $j = 1, \dots, 2n$, apply $b \circ (\alg I)^j$ to $(s, t)$; concatenate the resulting bits to get $u \in \bit{2n}$;
\item output $P_u \sigma P_u$.
\end{itemize}
\end{enumerate}
\end{scheme}

Correctness of the scheme is straightforward; fix a key-pair $(i, t)$, a randomly sampled $d \in D_i$, and the corresponding $r$. Then
$$
\Dec_t (\Enc_i(\rho))
= \Dec_t\bigl(\ket{f_i^{2n}(d)}\bra{f_i^{2n}(d)} \otimes P_r \rho P_r\bigr)
= P_u P_r \rho P_r P_u
= \rho\,,
$$
where the last step follows from the fact that $u = r$ for valid ciphertexts. It remains to show that this scheme is secure against chosen-plaintext attacks. We begin by proving indistinguishability of ciphertexts for the quantum one-time pad which uses randomness supplied by a qPRG. We first set the following notation. Recall from \expref{Section}{sec:quantum-prelims} that a string $r$ of $2n$ bits determines a Pauli group element $P_r \in U(2^n)$. Given an $n$-qubit register $A$, an arbitrary register $B$, and $\rho \in \states(\H_A \otimes H_B)$, define $\mathbb P_{r;A} (\rho) := (P_r \otimes \one_B) \rho (P_r \otimes \one_B)$.

\begin{lemma}\label{lemma:indistinguish}
Suppose $G:\bit{n}\rightarrow\bit{m}$ is a qPRG.  Then for any efficiently preparable states $\rho_{AB} \in \states(\H_A \otimes \H_B)$ and $\sigma_A \in \states({\H_A})$, and any QPT~$\mathcal D$,
\begin{equation} \label{eq:prgindisti}
\Bigg|\underset{s \inrand \bit{n}}{\pr}\bigg[\mathcal{D} (\mathbb P_{G(s); A}(\rho_{AB})) = 1\bigg]
-\underset{s \inrand \bit{n}}{\pr}\bigg[\mathcal D(\mathbb P_{G(s); A}(\sigma_A\otimes \rho_B)) = 1\bigg]\Bigg| \leq \negl(n)\,.
\end{equation}
\end{lemma}
\begin{proof}
The two key observations are (i.) distinguishability as in Equation~\eqref{eq:prgindisti} is impossible if we replace $G(s)$ with uniform randomness, and (ii.) with only classical input/output access to $G$, we can simulate $\mathcal D(\mathbb P_{G(s); A} (\cdot))$ 
Putting these two facts together, it follows that achieving \eqref{eq:prgindisti} implies that outputs of $G$ can be distinguished from uniformly random.

Formally, let us assume that there is an adversary $\mathcal D$ that violates our hypothesis, i.e., that distinguishes some pair of inputs $(\mathbb P_{G(s); A}(\rho_{AB}),\, \mathbb P_{G(s); A}(\sigma_A\otimes \rho_B))$ with probability at least $1/p(n)$ for some polynomial $p$.  Then we'll show an algorithm $\mathcal D'$, that breaks the pseudorandom generator~$G$. On input $y\in\{0,1\}^m$, algorithm $\mathcal{D}'$ does the following:
\begin{itemize}
\item with probability $1/2$, run $\mathcal D$ on input $\mathbb P_{y;A}(\rho_{AB})$;
\item with probability $1/2$, run $\mathcal D$ on input $\mathbb P_{y;A}(\sigma_A\otimes \rho_B)$.
\end{itemize}
Now if $\mathcal{D}$ is able to correctly determine which of the cases we gave it, $\mathcal{D}'$ decides that $y$ must have been distributed pseudorandomly and outputs $1$, else it decides that $y$ is uniformly distributed and outputs $0$.

Notice that if $y=G(s)$, by definition $\mathcal D'$ outputs $1$ when $\mathcal D$ correctly distinguishes the two inputs, which occurs with probability at least $1/2+1/p(n)$ by the assumption on $\mathcal D$. On the other hand, suppose $y\inrand\bit{m}$; then the register $A$ is mapped to the maximally mixed state, and hence $\mathbb P_{y;A}(\rho_{AB}) = \mathbb P_{y;A}(\sigma_A\otimes \rho_B) = \one_A \otimes \rho_B$. In that case, $\mathcal D$ is correct with probability at most $1/2+\negl(n)$ (indeed, this is true for any QPT.) We conclude that $\mathcal D'$ distinguishes the case $y = G(s)$ from the case $y \inrand \bit{m}$ with non-negligible probability; this contradicts the assumption that $G$ is a qPRG.
\qed\end{proof}

Finally, to prove that the construction in \expref{Scheme}{def:PKE-from-qTOWF} is IND-CPA-secure, and thus establish \expref{Theorem}{thm:intro:PKE}, it remains to extend the above proof to a slightly more general scenario. Recall that $\Enc_i(\rho)= \outerprod{f_i^{2n}(d)}{f_i^{2n}(d)} \otimes P_r \rho P_r$ where $r = G(d)$. \expref{Lemma}{lemma:indistinguish} already shows that essentially no QPT adversary can distinguish $(P_{r}\otimes \one_{E}) \rho_{ME} (P_{r}\otimes \one_{E})$ from $(P_{r}\otimes \one_{E})(\ket{0}\bra{0}\otimes \rho_E)(P_{r}\otimes \one_{E})$, for any efficiently preparable bipartite state $\rho_{ME}$ over the message space and the environment. It remains to show that this indistinguishability still holds if the adversary is also provided the classical advice $f_i^{2n}(d)$.  We can prove this extended indistinguishability by extending the hybrid argument in the proof of \expref{Lemma}{lemma:prg-construction} in a standard way. To sketch the argument, first recall that the ``predictor'' algorithm succeeds at predicting the $i+1$\textsuperscript{st} bit of $G(U_n)$ given as input the first $i$ bits of the output of $G$. Now we also allow the predictor to read the bits of $f_i^{2n}(d)$. Success implies breaking the hard-core of $f$ (which is used to define and ensure the security of the qPRG $G$). We conclude that the states
$$
\outerprod{f_i^{2n}(d)}{f_i^{2n}(d)} \otimes \mathbb P_{G(s);M} (\rho_{ME})
\quad \text{and} \quad
\outerprod{f_i^{2n}(d)}{f_i^{2n}(d)} \otimes \mathbb P_{r';M} (\rho_{ME})
$$
are computationally indistinguishable for uniformly random $s, r'$. The right-hand side encryption above obviously satisfies IND-CPA, so we also have  computational indistinguishability of
$$
\outerprod{f_i^{2n}(d)}{f_i^{2n}(d)} \otimes \mathbb P_{r';M} (\rho_{ME})
\quad \text{and} \quad
\outerprod{f_i^{2n}(d)}{f_i^{2n}(d)} \otimes \mathbb P_{r';M} (\egoketbra{0}_M \otimes \rho_E)\,.
$$
By transitivity of computational indistinguishability, we conclude that
$$
\outerprod{f_i^{2n}(d)}{f_i^{2n}(d)} \otimes \mathbb P_{G(s);M} (\egoketbra{0}_M \otimes \rho_E)
\quad \text{and} \quad
\outerprod{f_i^{2n}(d)}{f_i^{2n}(d)} \otimes \mathbb P_{G(s);M} (\rho_{ME})\,,
$$
which completes the proof of \expref{Theorem}{thm:intro:PKE}.



\section{Conclusion}\label{sec:conclusions}

We have defined semantic security for the encryption of quantum data and shown its equivalence with indistinguishability; these results are given in the uniform model for quantum computations, but as is standard classically (see Chapter~5 of Goldreich's text~\cite{Goldreich2004}), these definitions can be adjusted to the case of ``non-uniform'' (but still polynomial-time) adversaries, whose messages need not be generated efficiently. While the proof is  somewhat different, the equivalence of IND and SEM still hold in this case. The constructions of encryption schemes (IND-CCA1 symmetric-key and IND-CPA public-key) presented above carry over as well, except that we now require primitives (qPRFs and qTOWPs, respectively) which are secure against non-uniform adversaries.

\subsection{Extensions and Future Work}
\label{sec:Future-Work}
We now briefly discuss some possible extensions of the above results. In most cases, these extensions are a matter of modifying our definitions and proofs in a fairly straightforward way. We leave the other cases as interesting open problems.

\begin{itemize}

\item Our definitions of IND-CPA, IND-CCA1 and SEM assume that all of the relevant messages are generated in polynomial time. In other words, our results assume ``uniform'' adversaries. As is standard classically (see Chapter 5 of Goldreich's text~\cite{Goldreich2004}), these definitions can be adjusted to the case of ``non-uniform'' (but still polynomial-time) adversaries, whose messages need not be generated efficiently. While the proof is of course somewhat different, the equivalence of IND and SEM still hold in this case. The encryption schemes (IND-CCA1 symmetric-key and IND-CPA public-key) presented above carry over as well, except that we now require primitives (qPRFs and qTOWPs, respectively) which are secure against non-uniform adversaries.

\item Our symmetric-key encryption scheme assumes that the decryption algorithm measures a portion of the input in order to recover a classical randomness string, prior to decrypting. One might find this requirement suspicious, e.g., if a perfect measurement device is too much to assume. This requirement can be removed, but we then need to assume that the relevant primitives (OWFs and qPRFs) are secure against superposition queries. This can also be achieved (see~\cite{Zhandry2012}).

\item One outstanding open problem is to define and construct schemes for CCA2 (adaptive chosen ciphertext attack) security in the case of the encryption of quantum states. Classically, CCA2 security is defined as CCA1, with the further property that the adversary is allowed to query the decryption oracle even {\em after} the challenge query, {\em provided} he does not query about the challenge ciphertext itself (otherwise the challenger aborts the game.) The obvious way to define this in the quantum world is to require that every decryption query performed by the adversary after the challenge query is `very different' from the challenge query itself (e.g., it is orthogonal to the challenge ciphertext.) But the problem here is that this condition might be impossible for the challenger to check: for example, the adversary might embed in a decryption query a component non-orthogonal to the challenge query, but with such a small amplitude that the challenger cannot detect it with high probability. Even if it is unclear whether this issue could raise problems in any actual reduction, it would be anyway a striking asimmetry to the classical case, because there would be no way for the challenger to check that the adversary actually fulfilled the required condition. Hence, giving a satisfactory definition for CCA2 security in the quantum world remains an interesting open problem.

\end{itemize}

\ifthenelse{\boolean{SUBMISSION}}
{
}
{
\subsection{Acknowledgements}

G.\,A. was supported by a Sapere Aude grant of the Danish Council for Independent Research, the ERC Starting Grant ``QMULT'' and the CHIST-ERA project ``CQC''. A.\,B. was supported by Canada's NSERC. B.\,F. was supported by the Department of Defense. T.\,G. was supported by the German Federal Ministry of Education and Research (BMBF) within EC-SPRIDE and CROSSING. C.\,S. was supported by a 7th framework EU SIQS and a NWO VIDI grant.  M.\,S. was supported by the Ontario Ontario Graduate Scholarship Program.
 T.\,G. and C.\,S. would like to thank COST Action IC1306 for networking support. A.\,B., G.\,A., T.\,G., and C.\,S. would like to thank the organizers of the Dagstuhl Seminar 15371 ``Quantum Cryptanalysis'' for providing networking and useful interactions and support during the writing of this paper.
}


\bibliographystyle{alpha}
\newcommand{\etalchar}[1]{$^{#1}$}


\begin{thebibliography}{AMTdW00}

\bibitem[Aar09]{Aar2009}
Scott Aaronson.
\newblock Quantum copy-protection and quantum money.
\newblock In {\em Computational Complexity, 2009. CCC'09. 24th Annual IEEE
  Conference on}, pages 229--242. IEEE, 2009.

\bibitem[ABB{\etalchar{+}}14]{ABB+2014}
Romain All{\'e}aume, Cyril Branciard, Jan Bouda, Thierry Debuisschert, Mehrdad
  Dianati, Nicolas Gisin, Mark Godfrey, Philippe Grangier, Thomas L{\"a}nger,
  Norbert L{\"u}tkenhaus, Christian Monyk, Philippe Painchault, Momtchil Peev,
  Andreas Poppe, Thomas Pornin, John Rarity, Renato Renner, Gregoire Ribordy,
  Michel Riguidel, Louis Salvail, Andrew Shields, Harald Weinfurter, and Anton
  Zeilinger.
\newblock Using quantum key distribution for cryptographic purposes: A survey.
\newblock {\em Theoretical Computer Science}, 560:62--81, 2014.

\bibitem[AC02]{AdcockCleve2002}
Mark Adcock and Richard Cleve.
\newblock A quantum goldreich-levin theorem with cryptographic applications.
\newblock In Helmut Alt and Afonso Ferreira, editors, {\em STACS 2002}, volume
  2285 of {\em Lecture Notes in Computer Science}, pages 323--334. Springer
  Berlin Heidelberg, 2002.

\bibitem[AC12]{AC2012}
Scott Aaronson and Paul Christiano.
\newblock Quantum money from hidden subspaces.
\newblock In {\em Proceedings of the forty-fourth annual ACM symposium on
  Theory of computing}, pages 41--60. ACM, 2012.

\bibitem[AKN98]{AKN1998}
Dorit Aharonov, Alexei Kitaev, and Noam Nisan.
\newblock Quantum circuits with mixed states.
\newblock In {\em Proceedings of the thirtieth annual ACM symposium on Theory
  of computing}, pages 20--30. ACM, 1998.

\bibitem[AMTdW00]{AMTW00}
Andris Ambainis, Michele Mosca, Alain Tapp, and Ronald de~Wolf.
\newblock Private quantum channels.
\newblock In {\em Foundations of Computer Science, 2000. Proceedings. 41st
  Annual Symposium on}, pages 547--553, 2000.

\bibitem[BB84]{BennettB84}
Charles Bennett and Gilles Brassard.
\newblock Quantum cryptography: Public key distribution and coin tossing.
\newblock In {\em Proceedings of the International Conference on Computers,
  Systems, and Signal Processing}, pages 175--179, 1984.

\bibitem[BBD09]{BBD09}
Daniel~J Bernstein, Johannes Buchmann, and Erik Dahmen, editors.
\newblock {\em Post-Quantum Cryptography}.
\newblock Springer, 2009.

\bibitem[BDF{\etalchar{+}}11]{BDF+11}
Dan Boneh, {\"{O}}zg\"{u}r Dagdelen, Marc Fischlin, Anja Lehmann, Christian
  Schaffner, and Mark Zhandry.
\newblock Random oracles in a quantum world.
\newblock In {\em Advances in Cryptology---ASIACRYPT 2011}, pages 41--69, 2011.

\bibitem[BFK09]{BFK2009}
Anne Broadbent, Joseph Fitzsimons, and Elham Kashefi.
\newblock Universal blind quantum computation.
\newblock In {\em Foundations of Computer Science, 2009. FOCS'09. 50th Annual
  IEEE Symposium on}, pages 517--526. IEEE, 2009.

\bibitem[BGS13]{BGS2013}
Anne Broadbent, Gus Gutoski, and Douglas Stebila.
\newblock Quantum one-time programs.
\newblock In {\em Advances in Cryptology--CRYPTO 2013}, pages 344--360.
  Springer, 2013.

\bibitem[BJ15]{BJ15}
Anne Broadbent and Stacey Jeffery.
\newblock {Quantum homomorphic encryption for circuits of low {$T$}-gate
  complexity}.
\newblock In {\em Crypto 2015}, pages 609--629, 2015.

\bibitem[BOCG{\etalchar{+}}06]{BCG+2006}
Michael Ben-Or, Claude Cr{\'e}peau, Daniel Gottesman, Avinatan Hassidim, and
  Adam Smith.
\newblock Secure multiparty quantum computation with (only) a strict honest
  majority.
\newblock In {\em Foundations of Computer Science, 2006. FOCS'06. 47th Annual
  IEEE Symposium on}, pages 249--260. IEEE, 2006.

\bibitem[BR03]{BR2003}
P.~Oscar Boykin and Vwani Roychowdhury.
\newblock Optimal encryption of quantum bits.
\newblock {\em Physical review A}, 67(4):042317, 2003.

\bibitem[Bro15]{Bro2015}
Anne Broadbent.
\newblock Delegating private quantum computations.
\newblock {\em Canadian Journal of Physics}, pages 941--946, Jun 2015.

\bibitem[BS16]{BS16}
Anne Broadbent and Christian Schaffner.
\newblock Quantum cryptography beyond quantum key distribution.
\newblock {\em Designs, Codes and Cryptography}, 78:351--382, 2016.

\bibitem[BZ13]{Boneh2013a}
Dan Boneh and Mark Zhandry.
\newblock {Secure Signatures and Chosen Ciphertext Security in a Quantum
  Computing World}.
\newblock In Ran Canetti and Juan~A. Garay, editors, {\em Crypto 2013}, volume
  8043 of {\em LNCS}, pages 361--379. Springer, 2013.

\bibitem[DD10]{DD10}
Simon~Pierre Desrosiers and Frédéric Dupuis.
\newblock Quantum entropic security and approximate quantum encryption.
\newblock {\em IEEE Transactions on Information Theory}, 56(7):3455--3464,
  2010.

\bibitem[Des09]{D09}
Simon~Pierre Desrosiers.
\newblock Entropic security in quantum cryptography.
\newblock {\em Quantum Information Processing}, 8(4):331--345, August 2009.

\bibitem[DH76]{hellman1976new}
Whitfield Diffie and Martin~E. Hellman.
\newblock New directions in cryptography.
\newblock {\em IEEE transactions on Information Theory}, 22(6):644--654, 1976.

\bibitem[DNS10]{DNS2010}
Fr{\'e}d{\'e}ric Dupuis, Jesper~Buus Nielsen, and Louis Salvail.
\newblock Secure two-party quantum evaluation of unitaries against specious
  adversaries.
\newblock In {\em Advances in Cryptology--CRYPTO 2010}, pages 685--706.
  Springer, 2010.

\bibitem[DNS12]{DNS2012}
Fr{\'e}d{\'e}ric Dupuis, Jesper~Buus Nielsen, and Louis Salvail.
\newblock Actively secure two-party evaluation of any quantum operation.
\newblock In {\em Advances in Cryptology--CRYPTO 2012}, pages 794--811.
  Springer, 2012.

\bibitem[FKS{\etalchar{+}}13]{FKS+2013}
Serge Fehr, Jonathan Katz, Fang Song, Hong-Sheng Zhou, and Vassilis Zikas.
\newblock Feasibility and completeness of cryptographic tasks in the quantum
  world.
\newblock In {\em Theory of Cryptography}, pages 281--296. Springer, 2013.

\bibitem[GGM86]{Goldreich1986}
Oded Goldreich, Shafi Goldwasser, and Silvio Micali.
\newblock How to construct random functions.
\newblock {\em Journal of the ACM}, 33(4):792--807, 1986.

\bibitem[GHS15]{Gagliardoni2015}
Tommaso Gagliardoni, Andreas H\"ulsing, and Christian Schaffner.
\newblock Semantic security and indistinguishability in the quantum world,
  2015.
\newblock \url{http://arxiv.org/abs/1504.05255}.

\bibitem[GL89]{GoldreichLevin89}
O.~Goldreich and L.~A. Levin.
\newblock A hard-core predicate for all one-way functions.
\newblock In {\em Proceedings of the Twenty-first Annual ACM Symposium on
  Theory of Computing}, STOC '89, pages 25--32, New York, NY, USA, 1989. ACM.

\bibitem[GM84]{GM}
Shafi Goldwasser and Silvio Micali.
\newblock Probabilistic encryption.
\newblock {\em Journal of Computer and System Sciences}, 28(2):270 -- 299,
  1984.

\bibitem[Gol04a]{Goldreich2004}
Oded Goldreich.
\newblock {\em Foundations of Cryptography: Volume 2, Basic Applications}.
\newblock Cambridge University Press, Cambridge, UK, 2004.

\bibitem[Gol04b]{Goldreich}
Oded Goldreich.
\newblock {\em Foundations of Cryptography: Volume 2, Basic Applications}.
\newblock Cambridge University Press, New York, NY, USA, 2004.

\bibitem[GPV08]{GPV08}
Craig Gentry, Chris Peikert, and Vinod Vaikuntanathan.
\newblock Trapdoors for hard lattices and new cryptographic constructions.
\newblock In {\em Proceedings of the Fortieth Annual ACM Symposium on Theory of
  Computing}, STOC '08, pages 197--206, New York, NY, USA, 2008. ACM.

\bibitem[HILL99]{HAstad:1999:PGA:312173.312213}
Johan H{\aa}stad, Russell Impagliazzo, Leonid~A. Levin, and Michael Luby.
\newblock A pseudorandom generator from any one-way function.
\newblock {\em SIAM J. Comput.}, 28:1364--1396, March 1999.

\bibitem[HLSW04]{HLSW2004}
Patrick Hayden, Debbie Leung, Peter~W Shor, and Andreas Winter.
\newblock Randomizing quantum states: Constructions and applications.
\newblock {\em Communications in Mathematical Physics}, 250(2):371--391, 2004.

\bibitem[KK07]{KK07}
Elham Kashefi and Iordanis Kerenidis.
\newblock Statistical zero knowledge and quantum one-way functions.
\newblock {\em Theoretical Computer Science}, 378(1):101 -- 116, 2007.

\bibitem[Kos07]{sim2}
Takeshi Koshiba.
\newblock Security notions for quantum public-key cryptography.
\newblock {\em IEICE TRANSACTIONS on Fundamentals of Electronics,
  Communications and Computer Sciences}, J90-A(5):367--375, Feb 2007.

\bibitem[Leu02]{L02}
Debbie~W. Leung.
\newblock Quantum vernam cipher.
\newblock {\em Quantum Information and Computation}, 2(1):14--34, 2002.

\bibitem[MRV07]{MRV07}
C.~{Moore}, A.~{Russell}, and U.~{Vazirani}.
\newblock {A classical one-way function to confound quantum adversaries}.
\newblock {\em eprint arXiv:quant-ph/0701115}, January 2007.

\bibitem[MS10]{MS2010}
Michele Mosca and Douglas Stebila.
\newblock Quantum coins.
\newblock {\em Error-Correcting Codes, Finite Geometries and Cryptography},
  523:35--47, 2010.

\bibitem[OTU00]{OKS00}
Tatsuaki Okamoto, Keisuke Tanaka, and Shigenori Uchiyama.
\newblock Quantum public-key cryptosystems.
\newblock In Mihir Bellare, editor, {\em Advances in Cryptology CRYPTO 2000},
  volume 1880 of {\em Lecture Notes in Computer Science}, pages 147--165.
  Springer Berlin Heidelberg, 2000.

\bibitem[PW08]{PW08}
Chris Peikert and Brent Waters.
\newblock Lossy trapdoor functions and their applications.
\newblock In {\em Proceedings of the Fortieth Annual ACM Symposium on Theory of
  Computing}, STOC '08, pages 187--196, New York, NY, USA, 2008. ACM.

\bibitem[Sha49]{Shannon1949}
C.~E. Shannon.
\newblock Communication theory of secrecy systems*.
\newblock {\em Bell System Technical Journal}, 28(4):656--715, Oct 1949.

\bibitem[Sho94]{Sho94}
Peter~W. Shor.
\newblock {Algorithms for Quantum Computation: Discrete Logarithms and
  Factoring}.
\newblock In {\em FOCS 1994}, pages 124--134. IEEE Computer Society Press,
  1994.

\bibitem[Son14]{Son2014}
Fang Song.
\newblock A note on quantum security for post-quantum cryptography.
\newblock In {\em Post-Quantum Cryptography}, pages 246--265. Springer, 2014.

\bibitem[Unr10]{Unr2010}
Dominique Unruh.
\newblock Universally composable quantum multi-party computation.
\newblock In {\em Advances in Cryptology--EUROCRYPT 2010}, pages 486--505.
  Springer, 2010.

\bibitem[Unr14]{Unr2014}
Dominique Unruh.
\newblock Revocable quantum timed-release encryption.
\newblock In {\em Advances in Cryptology--EUROCRYPT 2014}, pages 129--146.
  Springer, 2014.

\bibitem[Unr15]{Unr15}
Dominique Unruh.
\newblock Non-interactive zero-knowledge proofs in the quantum random oracle
  model.
\newblock In {\em Advances in Cryptology---EUROCRYPT 2015}, pages 755--784,
  2015.

\bibitem[Vel13]{Velema13}
Maria Velema.
\newblock Classical encryption and authentication under quantum attacks.
\newblock Master's thesis, Master of Logic, University of Amsterdam, 2013.
\newblock \url{http://arxiv.org/abs/1307.3753}.

\bibitem[Wie83]{Wie1983}
Stephen Wiesner.
\newblock Conjugate coding.
\newblock {\em ACM Sigact News}, 15(1):78--88, 1983.

\bibitem[WZ82]{Wootters1982}
W.~K. Wootters and W.~H. Zurek.
\newblock A single quantum cannot be cloned.
\newblock {\em Nature}, 299(5886):802--803, Oct 1982.

\bibitem[XY12]{sim3}
Chong Xiang and Li~Yang.
\newblock Indistinguishability and semantic security for quantum encryption
  scheme.
\newblock {\em Proc. SPIE}, 8554:85540G--85540G--8, 2012.

\bibitem[Zha12]{Zhandry2012}
Mark Zhandry.
\newblock {How to Construct Quantum Random Functions}.
\newblock In {\em FOCS 2012}, pages 679--687. IEEE, 2012.

\end{thebibliography}

\appendix

\section{Alternative Definitions of Quantum Security}
\label{app:alternate-Sem-Sec}
Here, we present further definitions of quantum semantic security  and indistinguishability, and prove their equivalence. The full chain of equivalencies is given in Figure~\ref{fig:equiv-defns}.

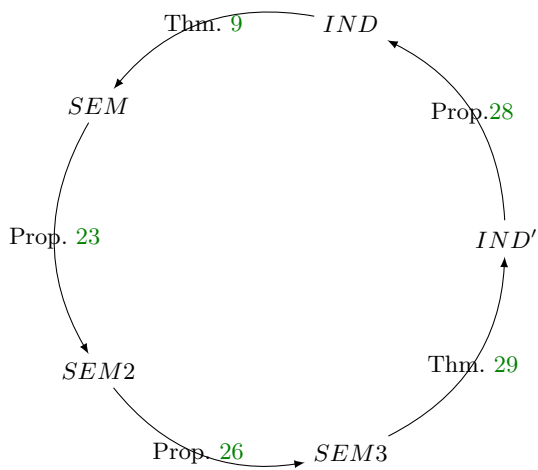
\begin{figure}
\begin{center}
\begin{tikzpicture}
\def \n {5}  
\def \radius {3cm}
\def \margin {15} 

\foreach \s [count=\i] in {IND',IND, SEM, SEM2, SEM3}
{
  \node (\i) at ({360/\n * (\i - 1)}:\radius) {$\s$};
}
 \draw[->, >=latex, bend right] (1) to node {Prop.\ref{prop:IND'-iff-IND}} (2);
  \draw[->, >=latex, bend right] (2) to node {Thm.~\ref{thm:ind''-implies-sem}} (3);
   \draw[->, >=latex, bend right] (3) to node {Prop.~\ref{prop:SEM-implies-SEM2}} (4);
    \draw[->, >=latex, bend right] (4) to node {Prop.~\ref{prop:SEM2-implies-SEM3}} (5);
     \draw[->, >=latex, bend right] (5) to node {Thm.~\ref{thm:SEM3-implies-IND'}} (1);
\end{tikzpicture}
\caption{Relationship between security definitions. }
\label{fig:equiv-defns}
\end{center}
\end{figure}

\subsection{SEM2}

\begin{definition}[Message-classical function generator] A  \textit{message-classical function generator} $\M$ is a QPT message generator (as in~SEM (\expref{Definition}{def:SEM})) such that for each $pk \in \K_{pub}$ and $\rho$ in the range of $\M(pk)$, there is some binary string $y$ such that $\ket{y} \in \H_F$ and $\rho_{MEF} = \rho_{ME} \otimes \egoketbra{y}$.
\end{definition}

That is, the $F$ system is classical, unentangled from and uncorrelated with the rest of $\rho$.

In particular, $\rho_F = \egoketbra{y}$.

\begin{definition}[SEM2]\label{def:SEMtwo}
A qPKE scheme $(\KeyGen, \Enc, \Dec)$ is \emph{SEM2-secure} if for any QPT adversary $\A$, there exists a QPT
simulator $\S$ such that for all message-classical function generators $\M$,
\begin{equation*}
\left|\Pr \big[ \; \A\big\{ (\Enc_{pk} \otimes \one_E)\rho_{ME} \big\} = \rho_{F} \; \big] -
\Pr \big[ \; \S (\rho_{E}) = \rho_{F} \; \big] \right| \leq \negl(n)
\end{equation*}
where the outputs of $\A$ and $\S$ are measured in the computational basis before equality is checked, $\rho_{MEF} \from \M(pk)$, and the probabilities are taken over $(pk, sk) \leftarrow \KeyGen(1^n)$ and the internal randomness of \Enc, $\A$, $\S$ and $\D$.
\begin{itemize}
\item \textbf{\emph{SEM2-CPA:}} In addition to the above, all QPTs are given oracle access to $\Enc_{pk}$.
\item \textbf{\emph{SEM2-CCA1:}} In addition to SEM2-CPA, $\M$ is given oracle access to $\Dec_{sk}$.
\end{itemize}
\end{definition}

\begin{proposition}
\label{prop:SEM-implies-SEM2}
If a quantum encryption scheme $(\KeyGen, \Enc, \Dec)$ is semantically secure, then it is SEM2 secure.
\end{proposition}

\begin{proof}
A message-classical function generator is also a message. In SEM, have the distinguisher $\D$ implement an equality test (simulate any efficient classical circuit implementing it; if the input lengths aren't the same, output 0 immediately). \qed
\end{proof}

\subsection{SEM3}

\begin{definition}[Message Generator-Function Pair] A \textit{message generator-function pair} is a tuple $(\M, f)$, such that $\M$ is a QPT message generator (as in~IND (\expref{Definition}{def:IND})) and $f = (f_n)_n$ is a QPT algorithm, such that $f_{pk} := f_n(pk)$ is the description of a boolean circuit, for $pk \in \K_{pub}$, with the number of input bits to $f_{pk}$ equal to the number of measurement gates in the quantum circuit $\M_n$. In the symmetric-key scenario, $f_n$ has no input.
\end{definition}

\begin{definition}[SEM3]
 \label{def:SEMthree}
A qPKE scheme $(\KeyGen, \Enc, \Dec)$ is \emph{SEM3-secure} if for any QPT adversary $\A$, there exists a QPT simulator $\S$ such that for all message generator-function pairs $(\M, f)$,
\begin{equation*}
\left|\Pr \big[ \; \A\big\{ (\Enc_{pk} \otimes \one_E)\rho_{ME} \big\} = f_{pk}(x) \; \big] -
\Pr \big[ \; \S (\rho_{E}) = f_{pk}(x) \; \big] \right| \leq \negl(n)
\end{equation*}
where the outputs of $\A$ and $\S$ are measured in the computational basis before equality is checked, $\rho_{ME} \from \M(pk)$, $x$ is the string of measurement results generating $\rho_{ME}$, and the probabilities are taken over $(pk, sk) \leftarrow \KeyGen(1^n)$ and the internal randomness of \Enc, $\A$, $\S$ and $\D$.
\begin{itemize}
\item \textbf{\emph{SEM3-CPA:}} In addition to the above, all QPTs are given oracle access to $\Enc_{pk}$.
\item \textbf{\emph{SEM3-CCA1:}} In addition to SEM3-CPA, $\M$ is given oracle access to $\Dec_{sk}$.
\end{itemize}
\end{definition}

We note that $f_{pk}$ is a function of the random input and measurement results, which completely determine the state. Hence, if $f_{pk}$ is the identity for all $pk$, and it can be computed given a ciphertext, this means we can compute measurement results necessary to prepare the state. Simulating the message generator but selecting for the correct measurement results would allow the preparation of the same state again, although this is not in general efficient.

\begin{proposition}[SEM2 $implies$ SEM3]
\label{prop:SEM2-implies-SEM3}
If a quantum encryption scheme $(\KeyGen, \Enc, \Dec)$ is SEM2 secure, then it is SEM3 secure.
\end{proposition}

\begin{proof}
A message generator-function pair $(\M,f)$ can be turned into a message-classical function generator $\M$ by copying the measurement results $x_1, x_2, \dots, x_m$ after each measurement gate, and applying $f_{pk}$ to $x_1 x_2 \dots x_m$ and letting the result be the $F$ system. \qed
\end{proof}

\subsection{IND'}
\begin{definition}[IND']
\label{def:IND'}
A qPKE scheme $(\KeyGen, \Enc, \Dec)$ is {\em IND' secure} if for every QPT adversary $\A=(\M,\D)$ we have:
\begin{equation*}
\Pr \big[ \; \D\big\{ (\Enc_{pk} \otimes \one_E)\rho^{(b)}_{ME} \big\} = b \big] \leq \tfrac{1}{2} + \negl(n)
\end{equation*}
where $\rho_{ME} \from \M(pk)$, for $b$ a uniformly random bit, $\rho^{(1)}_{ME} = \rho_{ME}$ and $\rho^{(0)}_{ME} = \egoketbra{0}_M \otimes \rho_E$, and the probabilities are taken over $(pk, sk) \leftarrow \KeyGen(1^n)$, $b$ and the internal randomness of \Enc, $\M$, and $\D$.
\begin{itemize}
\item \textbf{\emph{IND'-CPA:}} In addition to the above, $\M$ and $\D$ are given oracle access to $\Enc_{pk}$.
\item \textbf{\emph{IND'-CCA1:}} In addition to IND'-CPA, $\M$ is given oracle access to $\Dec_{sk}$.
\end{itemize}
\end{definition}

\begin{proposition}[IND' $\iff$ IND]
\label{prop:IND'-iff-IND}
A quantum encryption scheme \\* $(\KeyGen, \Enc, \Dec)$ is IND' secure if and only if it is IND secure.
\end{proposition}

\begin{proof}

We drop brackets and the register subcripts where possible. 

\begin{eqnarray*}
&& \Pr[\D(\Enc_{pk} \otimes \one_E)\rho^{(b)} = b] \\
&& = \Pr[\D(\Enc_{pk} \otimes \one_E)\rho^{(b)} = b \ | \ b = 1]\Pr[b = 1] \\
&& \ \ \ + \Pr[\D(\Enc_{pk} \otimes \one_E)\rho^{(b)} = b \ | \ b = 0]\Pr[b = 0] \\
&& =  \tfrac{1}{2} (\Pr[\D(\Enc_{pk} \otimes \one_E)\rho = 1] + \Pr[\D(\Enc_{pk}\ket{0}\bra{0} \otimes \rho_E) = 0]) \\
&& \leq \tfrac{1}{2} (\Pr[\D(\Enc_{pk} \otimes \one_E)\rho = 1] + 1- \Pr[\D(\Enc_{pk}\ket{0}\bra{0} \otimes \rho_E) = 1]) \\
&& =  \tfrac{1}{2} +  \tfrac{1}{2}(\Pr[\D(\Enc_{pk} \otimes \one_E)\rho = 1] - \Pr[\D(\Enc_{pk}\ket{0}\bra{0} \otimes \rho_E) = 1])
\end{eqnarray*}

Note that we only get $\leq$ since $\D$ may output some binary string other than 0 or 1. So:

\begin{eqnarray*}
&& \Pr[\D(\Enc_{pk} \otimes \one_E)\rho^{(b)} = b] -  \tfrac{1}{2} \\
&& \leq \tfrac{1}{2}|\Pr[\D((\Enc_{pk} \otimes \one_E)\rho = 1] - \Pr[\D(\Enc_{pk}\ket{0}\bra{0} \otimes \rho_E) = 1]|\,.
\end{eqnarray*}
Thus, IND $\implies$ IND'.

Now consider replacing $\D$ with the distinguisher which starts the same as $\D$, but if $\D$ would have output something other than 0 or 1, it simply outputs 0. Then the quantity $|\Pr[\D((\Enc_{pk} \otimes \one_E)(\rho))) = 1] - \Pr[\D(\Enc_{pk}\ket{0}\bra{0} \otimes \rho_E) = 1]|$ is the same for this new distinguisher, so without loss of generality, $\D$ only outputs 0 or 1.

Then the first $\leq$ becomes an = , i.e.

\begin{eqnarray*}
&& \Pr[\D(\Enc_{pk} \otimes \one_E)\rho^{(b)} = b] -  \tfrac{1}{2} \\
&& = \tfrac{1}{2}(\Pr[\D(\Enc_{pk} \otimes \one_E)\rho = 1] - \Pr[\D(\Enc_{pk}\ket{0}\bra{0} \otimes \rho_E) = 1])
\end{eqnarray*}

and, similarly,

\begin{eqnarray*}
&& \Pr[\D(\Enc_{pk} \otimes \one_E)\rho^{(b)} = b \oplus 1] \\
&& = \Pr[\D(\Enc_{pk} \otimes \one_E)\rho^{(b)} = b  \oplus 1 \ | \ b = 1]\Pr[b = 1] \\
&& \ \ \ + \Pr[\D(\Enc_{pk} \otimes \one_E)\rho^{(b)} = b  \oplus 1 \ | \ b = 0]\Pr[b = 0] \\
&& =  \tfrac{1}{2} (\Pr[\D(\Enc_{pk} \otimes \one_E)\rho = 0] + \Pr[\D(\Enc_{pk}\ket{0}\bra{0} \otimes \rho_E) = 1]) \\
&& = \tfrac{1}{2} (1 - \Pr[\D(\Enc_{pk} \otimes \one_E)\rho = 1] + \Pr[\D(\Enc_{pk}\ket{0}\bra{0} \otimes \rho_E) = 1]) \\
&& =  \tfrac{1}{2} +  \tfrac{1}{2}(\Pr[\D(\Enc_{pk}\ket{0}\bra{0} \otimes \rho_E) = 1] - \Pr[\D(\Enc_{pk} \otimes \one_E)\rho = 1])
\end{eqnarray*}

so,

\begin{eqnarray*}
&& \Pr[\D(\Enc_{pk} \otimes \one_E)\rho^{(b)} = b \oplus 1] - \tfrac{1}{2} \\
&& = \tfrac{1}{2}(\Pr[\D(\Enc_{pk}\ket{0}\bra{0} \otimes \rho_E) = 1] - \Pr[\D(\Enc_{pk} \otimes \one_E)\rho = 1])
\end{eqnarray*}

Combining the above,

\begin{eqnarray*}
&& \tfrac{1}{2}|\Pr[\D(\Enc_{pk} \otimes \one_E)\rho = 1] - \Pr[\D(\Enc_{pk}\ket{0}\bra{0} \otimes \rho_E) = 1]|  \\
&& = \max\{\Pr[\D(\Enc_{pk} \otimes \one_E)\rho^{(b)} = b] - \tfrac{1}{2}, \Pr[\D(\Enc_{pk} \otimes \one_E)\rho^{(b)} = b \oplus 1] - \tfrac{1}{2} \}
\end{eqnarray*}

Hence IND' $\implies$ IND by applying IND to both $\D$ and $\D \oplus 1$ (the latter outputs the answer opposite to $\D$), for $\Pr[\D(\Enc_{pk} \otimes \one_E)\rho^{(b)} = b]$ and $\Pr[\D(\Enc_{pk} \otimes \one_E)\rho^{(b)} = b \oplus 1]$, respectively. The maximum of two negligible functions is again negligible. \qed

\end{proof}

\begin{theorem}[SEM3 $\implies$ IND']
\label{thm:SEM3-implies-IND'}
If a quantum encryption scheme \\* $(\KeyGen, \Enc, \Dec)$ is SEM3 secure, then it is IND' secure.
\end{theorem}

\begin{proof}
We drop brackets and the register subcripts where possible. 

Let $(\M, \D)$ be an IND' adversary.

Let us consider the SEM3 message generator $\M'$ which runs $\rho_{ME} \from \M$ and outputs (with probability $\frac12$ each) either the state $\rho_{ME}$ or the state $\egoketbra{0}_M \otimes \rho_{E}$, and we denote its output by $\rho'_{ME}$ In particular, it prepares a random bit $b$ to do so by measuring a ancilary qubit to which the Hadamard was applied.

Define $f_{pk}(xb) = b$.

Define the SEM3 adversary $\A := \D$.

In this way, the SEM3 game simulates the indistinguishability game, and
\[ \Pr[\A((\Enc_{pk} \otimes \one_E)(\rho'_{ME}))) = f_{pk}(xb)] = \Pr[\D((\Enc_{pk} \otimes \one_E)(\rho^{(b)}))) = b] \]

Now, by SEM3, there is some simulator $\S$ for $\A$ so that

\[ |\Pr[\A(\Enc_{pk} \otimes \one_E)\rho'_{ME} = f_{pk}(xb)] - \Pr[\S\rho'_E = f_{pk}(xb)]| \leq \negl(n) \]

i.e.

\[ |\Pr[\D((\Enc_{pk} \otimes \one_E)(\rho^{(b)}))) = b] - \Pr[\S\rho'_E = f_{pk}(xb)]| \leq \negl(n) \]

Note that $\S$'s input $\rho'_E$ is independent of $b$. Hence

\[\Pr[\S \rho'_E = f_{pk}(xb)] \leq \tfrac{1}{2} \, .\]

Finally, by the triangle inequality applied to the last two inequalities,

\[ \Pr[\D(\Enc_{pk} \otimes \one_E)\rho^{(b)} = b]  \leq  \tfrac{1}{2} + \negl(n) \] \qed

\end{proof}


\end{document}